\documentclass[11pt]{article}
\pdfoutput=1
\usepackage{style}
\usepackage{shortcuts}
\usepackage[margin=1in]{geometry}

\title{Recognizing Sumsets is NP-Complete}
\author{Amir Abboud\thanks{Weizmann Institute of Science and INSAIT, Sofia University ``St. Kliment Ohridski''. This work is part of the project CONJEXITY that has received funding from the European Research Council (ERC) under the European Union's Horizon Europe research and innovation programme (grant agreement No.~101078482). Supported by an Alon scholarship and a research grant from the Center for New Scientists at the Weizmann Institute of Science.} \and Nick Fischer\thanks{INSAIT, Sofia University ``St. Kliment Ohridski''. Partially funded by the Ministry of Education and Science of Bulgaria's support for INSAIT, Sofia University ``St. Kliment Ohridski" as part of the Bulgarian National Roadmap for Research Infrastructure.
Parts of this work were done while the author was at Weizmann Institute of Science.} \and Ron Safier\thanks{Weizmann Institute of Science.} \and Nathan Wallheimer\thanks{Weizmann Institute of Science.}}
\date{}

\begin{document}

\maketitle
\begin{abstract}
\noindent
Sumsets are central objects in additive combinatorics. In 2007, Granville asked whether one can efficiently recognize whether a given set $S$ is a sumset, i.e. whether there is a set $A$ such that $A+A=S$. Granville suggested an algorithm that takes exponential time in the size of the given set, but can we do polynomial or even linear time? This basic computational question is indirectly asking a fundamental structural question: do the special characteristics of sumsets allow them to be efficiently recognizable? In this paper, we answer this question negatively by proving that the problem is NP-complete. Specifically, our results hold for integer sets and over any finite field. Assuming the Exponential Time Hypothesis, our lower bound becomes~\smash{$2^{\Omega(n^{1/4})}$}.
\end{abstract}

\section{Introduction} \label{sec:introduction}
Additive combinatorics has been rising steadily towards becoming one of the most important branches of mathematics for (theoretical) computer scientists. Its insights have catalyzed major advances in diverse topics such as communication complexity~\cite{BenSassonLR14,KelleyLM24}, randomness extractors~\cite{BarakIW06}, coding theory~\cite{BhowmickL15}, property testing~\cite{Alon02,Samorodnitsky07}, graph sparsification~\cite{Hesse03,AbboudB17}, worst-case to average-case reductions~\cite{AsadiGGS22}, fine-grained hardness of approximation~\cite{JinX23,AbboudBF23}, and also in algorithm design for Subset-Sum~\cite{ChaimovichFG89,BringmannN20,BringmannW21}, Knapsack~\cite{ChenLMZ24,Bringmann24,Jin24}, $3$-SUM~\cite{ChanL15}, and sparse convolutions~\cite{BringmannN21}. The surveys~\cite{Trevisan09,Viola11,Bibak13,Lovett17} include many other references and applications.

The \emph{sumset} of a set $A$ is the set of all pairwise sums $A+A:= \set{a+b : a, b \in A}$, and a set~$S$ is called a sumset if there exists a set $A$ such that $A+A=S$. Sumsets are a central object in additive combinatorics; a field that is often defined (e.g.\ in Wikipedia) ``by example'' while referring to fundamental results relating the properties of sets to those of their sumsets.

In this work, we are interested in computational problems in additive combinatorics. Perhaps the first question one might ask is whether a given set is a sumset.

\begin{definition}[The Sumset Recognition Problem]
\label{def:sumset_recognition}
Given a set $S \subseteq \mathbb{G}$ over a group $\mathbb{G}$ of size $n$ (e.g. the integers $\{1,\ldots,n\}$ or $\Field_p^d$ where $p^d=n$) decide whether there exists a set $A \subseteq \mathbb{G}$ such that $S=A+A$.
\end{definition}

In 2007, Croot and Lev~\cite{CrootL07} published a paper enumerating dozens of open questions by experts in the additive combinatorics community. The only computational question on the list, suggested by Andrew Granville, concerns the Sumset Recognition problem (Problem 4.11, ``Recognizing sumsets algorithmically'').

\begin{goq}
Is there an efficient algorithm for the Sumset Recognition problem?
\end{goq}

A naive algorithm solves the problem in exponential $2^n \cdot n^2$ time by enumerating all $2^n$ subsets~\makebox{$A \subseteq \mathbb{G}$} of the group, computing their sumset $A+A$ in quadratic time, and checking whether it equals the given set $S$. When raising the question, Granville observed that one can limit the search space to subsets of the given set, which leads to a $2^{|S|}$ upper bound. But does there exist a much faster strategy? Is there a polynomial time, or even a linear time algorithm?

While being a purely computational question, Granville is asking a fundamental question about the structure of sumsets, as pointed out in a joint paper~\cite{AlonGU10} with Alon and Ubis: ``\emph{Perhaps (though this seems unlikely) any sumset contains enough structure that is quickly identifiable? Perhaps most non-sumsets are easily identifiable in that they lack certain structure?}'' 

Despite the piquing interest in additive combinatorics among computer scientists, including the breakthrough by computer scientists~\cite{KelleyM23} on a central (non-computational) question of the field,\footnote{Notably, the result by Kelley and Meka~\cite{KelleyM23} on $3$-term arithmetic-progression-free sets resolves the very first problem in the aforementioned list~\cite{CrootL07}.} Granville's question has remained open and was recently raised again in a MathOverflow post~\cite{mathoverflow20}.

One intuitive explanation for the difficulty in settling the complexity of Sumset Recognition is that it is a \emph{factoring}-type of problem. In \emph{the} Factoring problem, which is famously elusive of the P vs.\ NP-Hard classification, we are given an integer $x = p\cdot q$, where $p,q$ are unknown primes, and are asked to find $p,q$. Replacing the multiplication operation over integers with the sumset operation (or \emph{Minkowski sum}) over sets, we get a variant of our problem.\footnote{We get the variant where given $S=A+B$ we are asked to return $A',B' \neq \set{0}$ such that $A'+B'=S$. For conciseness, we focus this paper on the basic $A+A$ version of Granville's question (and on recognition rather than factoring). Intuitively, the other variants are only harder.} How does changing the operation affect the complexity? 

\subsection{Results}
This paper gives a negative answer to Granville's question, proving that the Sumset Recognition problem is NP-hard. This gives formal support for the intuition of Alon, Granville, and Ubis that sumsets are not structured enough to allow for fast identification.

\begin{theorem} \label{thm:np-hard-integer}
The Sumset Recognition problem over $\Int$ is NP-complete.
\end{theorem}

While the proof can be viewed as yet another ``gadget reduction'', its high-level structure and implementation level gadgets are quite novel and specific to the special structure of sumsets. In particular, we rely crucially on results from additive combinatorics such as constructions of \emph{Sidon sets}; while such tools have become commonplace in theoretical computer science (e.g. in the aforementioned areas) we are not aware of any previous NP-completeness reduction that uses them.\footnote{The related \emph{Behrend sets} have been used in some fine-grained hardness reductions from the 3-SUM problem~\cite{DudekGS20} and in a barrier for compressing SAT instances~\cite{DellM14} but these settings are very different.} We are hopeful that our reduction framework will spur more results on the complexity of additive combinatorics questions.

The proof is via a reduction from $3$-SAT on $N$ variables and clauses to Sumset Recognition on a set of $|S|=O(N^4)$ integers in an interval of size $|\mathbb{G}|=O(N^4)$. This not only proves the NP-hardness in terms of the parameter $n=|S|$ but also in terms of $n=|\mathbb{G}|$ as in Definition~\ref{def:sumset_recognition}, which is a stronger statement.

The next theorem shows that the reduction can be implemented with additions in any finite field; which are the most standard abelian groups in additive combinatorics.

\begin{theorem} \label{thm:np-hard-finite-field}
For any prime $p$, the Sumset Recognition problem over $\Field_p^d$ is NP-complete.
\end{theorem}

The proof uses a similar high-level approach to the one for integers, but the different properties of fields compared to integers (e.g. that there might be elements $x$ with $2x = 0$) cause significant new challenges at the implementation level.

\subsection{Related Work}
The question of covering a set $S$ by a sumset $A+A$ with an $A$ that has minimal size has received considerable attention in the literature, both mathematically where the minimal size to cover the integers $\{1,\ldots,n\}$ is not known~\cite{Moser60,BulteauFRV15}, and algorithmically where parameterized and approximation algorithms have been studied~\cite{FagnotFV09,BulteauFRV15}. A related problem that has arisen in applications related to radiation therapy is to cover a given set $S$ with the set of all subset sums of a set $X$ of minimal size~\cite{CollinsKSY07}. Both aforementioned problems were shown to be NP-Hard by reductions that are much less involved than ours, heavily utilizing that one is looking for the minimal set.

A very recent paper by Chen, Nadimpalli, Randolph, Servedio, and Zamir~\cite{ChenNRSZ24} studies the \emph{query} complexity of the same problem from Granville's question, in the \emph{property testing} model. They prove a lower bound of $\Omega(\sqrt{n})$ on the number of element queries to a set $S$ in an abelian group of size $n$ that are required to test whether $S$ is a sumset or $\eps$-far from being one. Needless to say, their (sublinear) result does not say whether the computational \emph{time} complexity of the problem is linear, polynomial, or exponential.

\subsection{Technical Overview} \label{sec:overview}
In this section, we give an intuitive overview of the new ideas leading to our results. We will mainly focus on the NP-hardness of the \emph{integer} Sumset Recognition problem, which already requires many of our insights. As may be expected, we prove NP-hardness by designing a polynomial-time reduction from the classical $3$-SAT\footnote{We quickly recall the 3-SAT problem: The input consists of a 3-CNF formula $\phi$ on $n$ variables $x_1, \dots, x_n$. That is, $\phi$ consists of $m$ \emph{clauses}, each of which consists of three \emph{literals} taking form $x_i$ or $\overline x_i$. An assignment $\alpha \in \set{0, 1}^n$ satisfies a positive literal $x_i$ if $\alpha_i = 1$ and a negative literal $\overline x_i$ if $\alpha_i = 0$. The assignment satisfies a clause if it satisfies at least one of its literals, and it satisfies~$\phi$ if it satisfies all clauses. The 3-SAT problem is to decide whether there is a satisfying assignment of $\phi$.} problem to Sumset Recognition. That is, given a $3$-SAT formula~$\phi$ on~$n$ variables, our task is to create a set $X_\phi \subseteq [\poly(n)]$ such that we can express $X_\phi = A + A$, for some set $A$, if and only if $\phi$ is satisfiable. 

We will present a conceptually natural approach that, as we will see later, requires significant technical work to work out. The baseline idea is to construct a set $X_\phi$ such that each feasible set~$A$, in the sense that $ A + A= X_\phi$, encodes an assignment. Then, in a second step, we will modify this construction in accordance with the input formula, such that only the sets $A$ corresponding to \emph{satisfying} assignments remain.

\paragraph{Step 1: Encoding Assignments}
Following this general idea, the first step is to construct a \emph{variable gadget,} which is a constant-sized set $V$ that can be expressed as the sumset of exactly two sets $V_0$ and $V_1$ (i.e., $V = V_0 + V_0 = V_1 + V_1$). We will interpret~$V_0$ as the $0$-assignment and $V_1$ as the $1$-assignment of a variable. For technical reasons that will become clear later, we also require that there are two elements~\makebox{$v_0 \in V_0 \setminus V_1$} and $v_1 \in V_1 \setminus V_0$.\footnote{Using a computer search, we have identified the smallest gadget that fits these requirements: $V = \set{0,\dots,18} \setminus \set{5}$, $V_0 = \set{0, 1, 2, 6, 7, 8, 9}$, $V_1 = \set{0, 1, 3, 6, 7, 8, 9}$, $v_0 = 2$ and $v_1 = 3$. However, equipped with the ``positioning'' and ``masking'' tools that we develop soon, we can find a significantly simpler and \emph{understandable} gadget. See \cref{sec:sat} for more details.} Then, we assemble $X_\phi$ to be a set of $n$ appropriately shifted copies of $V$:
\begin{equation*}
    X_\phi = \bigcup_{i \in [n]} (V + 2s_i).
\end{equation*}
The hope is that the sumset roots of $X_\phi$ are in one-to-one corresponds with assignments $\alpha \in \set{0, 1}^n$, and have the form 
\begin{equation*}
    A_\alpha = \bigcup_{i \in [n]} (V_{\alpha_i} + s_i).
\end{equation*}
This looks promising, as indeed the shifted set $V + 2s_i$ has two unique sumset roots, namely $V_0 + s_i$ and $V_1 + s_i$. However, a major problem with this approach is that $A_\alpha + A_\alpha$ also contains \emph{cross terms} of the form $s_i + s_j + V_{\alpha_i} + V_{\alpha_j}$, for $i \neq j$. These cross terms also have to be included into~$X_\phi$---but of course, we cannot ``hardcode'' any variable assignment~$V_{\alpha_i}$. That is, for each $i,j$ we need to add $s_i+s_j+V_{\alpha_i} + V_{\alpha_j}$ to $X_{\phi}$ without knowing whether each of $V_{\alpha_i}$ and $V_{\alpha_j}$ is $V_0$ or $V_1$. This puts us in a precarious situation: In order to build $X_\phi$ oblivious to the assignment, we would have to find a variable gadget $V$ with sumset roots~$V_0$ and $V_1$, i.e. $V=V_0+V_0=V_1+V_1$, which additionally satisfies that $V = V_0 + V_1$. If we could do this, we could simply add $s_i+s_j+V$ to $X_{\phi}$ for all $i,j$. Unfortunately, this seems too much to ask for, and we could not find such a gadget. Our actual solution for the cross terms is significantly more involved and uses non-trivial constructions from additive combinatorics. We will describe it towards the end of this overview; for now, let us ignore the problematic cross terms and continue with the general picture.

\paragraph{Step 2: Enforcing Satisfying Assignments}
The second step is to include the $m$ clauses into the construction of $X_\phi$ to make sure that only the roots $A_\alpha$ corresponding to \emph{satisfying} assignments~$\alpha$ survive. Our approach is as follows: Let $t_1, \dots, t_m$ denote a fresh set of shifts. To encode the clauses, we would like to force $A_\alpha$ to also include the sets $C_k$ defined as follows:
\begin{equation*}
    C_k = \bigcup_{\substack{x_i \text{ appears} \\ \text{positively in} \\ \text{the $k$-th clause.}}} \set{t_k - s_i - v_1} \cup \bigcup_{\substack{x_i \text{ appears} \\ \text{negatively in} \\ \text{the $k$-th clause.}}} \set{t_k - s_i - v_0}
\end{equation*}
Note that these sets depend on the formula $\phi$ but not on the assignment $\alpha$.
While in the reduction we control $X_{\phi}$ but not $A_{\alpha}$, we would like to add certain elements to $X_{\phi}$ that will force any feasible~$A_{\alpha}$ to include all $C_k$'s.
At the very least, we have to add all elements in $\bigcup_{k, \ell} (C_k + C_\ell)$ to $X_\phi$. In addition, and this is key, we insert all elements $t_1, \dots, t_m$ into $X_\phi$. Our intention is to force that $t_k \in A_\alpha + A_\alpha$ if and only if $\alpha$ satisfies the $k$-th clause. Indeed, if the $k$-th clause is satisfied by a positive literal~$x_i$~(i.e.,~$\alpha_i = 1$), then, recalling that $s_i+v_1 \in A_{\alpha}$, we have $t_k = (v_1 + s_i) + (t_k - s_i - v_1) \in A_\alpha + A_\alpha$. And similarly, if it is satisfied by a negative literal $\overline x_i$ (i.e., $\alpha_i = 0$) then since $s_i+v_0 \in A_{\alpha}$ we have that $t_k = (v_0 + s_i) + (t_k - s_i - v_0) \in A_\alpha + A_\alpha$.

\medskip
Ideally, this should complete the description of our reduction---however, due to various interesting challenges, the story does not end here. Specifically, we face two major problems, which we will describe, along with our solutions, in the following paragraphs. For both of these problems, we managed to find surprisingly clean, modular solutions in terms of the upcoming~\cref{lem:pos-int,lem:masking-int}.

\paragraph{Problem 1: Positioning}
The first issue in this outline is how to prove the ``completeness'' direction. I.e., given a set $A$ with $X_\phi = A + A$, can we guarantee that $\phi$ is satisfiable? If the set~$A$ takes the form~$A_\alpha$ from before, for some assignment $\alpha$, then this is plausible: We merely have to ensure that we can never spuriously express $t_j$ as the sum of two unrelated elements from $A_\alpha$. This can be achieved by choosing the set of shifts $\set{s_1, \dots, s_n, t_1, \dots, t_m}$ to be a (slightly modified) \emph{Sidon set}, i.e., a set without non-trivial solutions to the equation $a + b = c + d$. See \cref{sec:sat} for more details.

However, the more serious issue is that it is a priori not clear that $A$ looks like $A_\alpha$. Could it not happen that~$A$ looks entirely obscure? To rule this out, we conceptually would like to enforce that~$A$ includes only certain elements by requiring that $A \subseteq U$ for some prespecified set~$U$; we loosely refer to this constraint as ``positioning $A$''. Then we could strategically choose $U$ to be in line with the sets~$A_\alpha$---and only these sets---by
\begin{equation*}
    U = \bigcup_{i \in [n]} ((V_0 \cup V_1) + s_i) \cup \bigcup_{k \in [m]} C_k.
\end{equation*}
And indeed, we show that the superset condition~$A \subseteq U$ can be enforced by suffering only \emph{constant} blow-up in the universe size. Formally, we prove the following lemma: 

\begin{restatable}[Positioning for $\Int$]{lemma}{lemposint} \label{lem:pos-int}
Let $X, U \subseteq \set{0,\dots,n}$. There is a set~\makebox{$X' \subseteq \set{0, \dots,2^{24} n}$} that can be constructed in time $\poly(n)$ and satisfying that:
\begin{equation*}
    \exists A \subseteq U : X = A + A \qquad\text{if and only if}\qquad \exists A' \subseteq \Int : X' = A' + A'.
\end{equation*}
\end{restatable}

The idea behind the proof of \cref{lem:pos-int} is as follows: We first design a \emph{skeleton set} $I \subseteq \set{0,\dots,c}$ over some constant-size universe (in the lemma, $c = 2^{24}$). We require that (i) $I$ is \emph{primitive}, i.e.\ that its sumset $I + I$ cannot be expressed as $J + J$ for some other set $J$. And (ii),~$I$ shall contain a specific constellation of integers tailored towards the following approach.

In the reduction, we will partition $\set{0,\dots,cn}$ into $c$ successive length-$n$ intervals and use the skeleton set $I$ for determining which of these intervals should be non-empty in the set $X'$. Let $X'_i$ denote the part of $X'$ that falls into the $i$-th interval, and similarly define $A'_i$. Then (i) implies that by choosing $X'$ as a subset of the intervals associated to $I + I$ (i.e.,~\makebox{$X'_i \neq \emptyset$} if and only if~\makebox{$i \in I + I$}), we can ensure that any set $A'$ with $X' = A' + A'$ is a subset of the intervals associated to $I$ (i.e.,~\makebox{$A'_i \neq \emptyset$} if and only if~$i \in I$).

Next, in order to enforce that $A \subseteq U$ in addition to $A+A=X$, we will place copies of $X$,$U$,$\{0\}$, and $[2n]$ at appropriately chosen intervals in $X'$ that are informed by the property (ii) of the skeleton set $I$. This, in turn, will force us to appropriately place copies of $A,U,\{0\}$, and~$[n]$ in $A'$. In a bit more detail about property (ii), our skeleton set $I$ contains the constellation $a, u, s - a, s - u$ for some integers $a, u, s$ that act independently (i.e., that satisfy no non-trivial equations, e.g. their pairwise sums are distinct) and that cannot be represented as the pairwise sum of elements from~$I$. The insight is that by setting $X'_{2a} = X$ and $X'_{2s - 2a} = \set{0}$ and $X'_{s} = U$, we in turn force that~\makebox{$A'_a + A'_a = X$} and~\makebox{$A'_{s-a} = \set{0}$} and
\begin{equation*}
    (A'_a + A'_{s - a}) \cup (A'_u + A'_{s - u}) = A'_a \cup (A'_u + A'_{s-u}) = U.
\end{equation*}
The latter effectively enforces that $A'_a \subseteq U$; to get the equality the solution separately ensures that~\makebox{$(A'_u + A'_{s-u}) = U$}.
In particular, the set $A = A'_a$ satisfies exactly the two conditions $X = A + A$ and $A \subseteq U$. We omit further details here, and instead refer to \cref{sec:pos}. Let us finally emphasize that throughout, we have not optimized most constants and that the specific constant $2^{24}$ can likely be dramatically lowered.

\paragraph{Problem 2: Masking (And Stopping the Infinite Game)} 
The second major problem we have to deal with is the problematic cross terms we encountered before. Recall that the issue was that the set $A_\alpha$ as constructed before consists of several copies of variable and clause gadgets. In the sumset $A_\alpha + A_\alpha$, our intention is that certain variables and clause gadgets interact (in the sense that the sums they cause are meaningful). However, the set $A_\alpha + A_\alpha$ also contains the cross terms caused by variable-variable, variable-clause, and clause-clause combinations. And even for the relevant variable-clause combinations, we have to cover the elements $t_k \pm 1$, which might or might not be present. For concreteness, focus again on the cross term $V_{\alpha_i} + V_{\alpha_j} + s_i + s_j$ from before, caused by two variable gadgets. As already argued, we cannot simply include this set into~$X_\phi$ as this would force fixed values for the variables $i$ and $j$.

Instead, we follow a different idea, which we call ``masking'': We will add to $A_\alpha$ even more sets $M_1, \dots, M_s$ (so-called ``masks''). The goal is that each cross term, say, $V_{\alpha_i} + V_{\alpha_j} + s_i + s_j$, is contained in a sumset $M_a + M_b$. In the simplest form, we choose the masks to be singleton sets and hit all unwanted cross term elements one by one. This requires some bookkeeping to make sure that the masks do not interfere with the relevant positions, like $t_j$.

Disregarding these details, there is a fundamental conceptual flaw in this idea: By including masks into $A_\alpha$, we also introduce \emph{new} cross terms involving these masks, e.g.,~\makebox{$V_{\alpha_i} + s_i + M_a$}. To cover these new cross terms, we would have to introduce even more masks, which in turn cause more cross terms---and this game continues infinitely. Is this inherent, or is there any hope to salvage this idea and stop the infinite game?

Perhaps surprisingly, it turns out that by choosing the masks appropriately---singletons in the first iteration of the game and intervals in the second iteration---we can stop the game after two iterations. The insight is that, with some extra work, after the first iteration all cross terms become predictable in the sense that we can determine their minimum and maximum elements. Then, using that for any set $A \subseteq \set{0,\dots,n}$ we have $A + \set{0,\dots,n} = \set{\min(A), \dots, \max(A) + n}$, after the second iteration we can simply include all cross terms into $X_\phi$ without introducing new masks. The details of this idea turn out to be extremely technical.

In summary, we prove the following lemma:

\begin{restatable}[Masking for $\Int$]{lemma}{lemmaskingint} \label{lem:masking-int}
Let $X, Y, U \subseteq \set{0,\dots,n}$. There are sets \makebox{$X', U' \subseteq \set{0,\dots, 2^{17} n^2}$} that can be constructed in time~$\poly(n)$ and satisfying that:
\begin{equation*}
    \exists A \subseteq U : X \subseteq A + A \subseteq Y \qquad\text{if and only if}\qquad \exists A' \subseteq U' : X' = A' + A'.
\end{equation*}
\end{restatable}

This lemma should be understood as follows: The Sumset Recognition problem is equivalent to the problem of detecting the existence of a set $A$ such that $A + A$ does not have to precisely match $X$. Rather, there are positions $Y \setminus X$ that we don't care about (and that are masked in the reduction). See \cref{sec:masking} for more details and the proof of \cref{lem:masking-int}.

\paragraph{Additional Challenges for \smash{\boldmath$\Field_p^d$}}
In contrast to many other settings in additive combinatorics, in our context, dealing with integers is simpler than dealing with, say,~\smash{$\Field_2^d$} or~\smash{$\Field_3^d$}. Even seemingly simple tasks like the construction of primitive sets are significantly more involved for \smash{$\Field_2^d$} and~\smash{$\Field_3^d$} (see \cref{sec:pos:sec:Fp}). This is mainly, but not exclusively, due to a special role that~\smash{$\Field_2^d$} takes here: In contrast to the integers and finite fields of odd order, there are order-$2$ elements $a$ (i.e., satisfying that $2a = 0$); in fact all nonzero elements have order $2$. This has several complicating consequences: For instance, there are no sets with a unique sumset representation, as whenever $X = A + A$ then also~\makebox{$X = (A + a) + (A + a)$}. Another complication is that in the positioning step we relied on the fact that $2a \neq 0$; over $\Field_2^d$ we have to take a different route which leads to weaker positioning and masking constructions (see \cref{lem:pos-Fp,lem:masking-Fp}) where we replace the conditions that $X = A + A$ and~\makebox{$X \subseteq A + A \subseteq Y$} by respectively $X = A + B$ and $X \subseteq A + B \subseteq Y$. Nevertheless, even these weaker versions are ultimately sufficient to prove NP-hardness.

\subsection{Conclusions and Open Questions}
In this paper, we have proven the NP-completeness of what may be considered the most basic computational question related to additive combinatorics. This opens up the door to investigating the time complexity of many other problems related to sumsets. Let us conclude with three specific questions that we find interesting.

First, while we have shown that it is not possible to recognize sumsets \emph{quickly}, it is still open to determine the extent to which non-trivial algorithmic strategies are helpful. In other words, we would like to know the fine-grained complexity of the Sumset Recognition problem. Our reduction takes a $3$-SAT instance on $n$ variables and produces a set of integers in an interval of size $O(n^4)$. This leads to a conditional~\smash{$2^{\Omega(n^{1/4})}$} lower bound; could there be a sub-exponential~\smash{$2^{O(n^{0.99})}$} algorithm?

\begin{theorem}[ETH-Hardness over $\Int$] \label{thm:eth-hard-integer}
The Sumset Recognition problem over $\Int$ cannot be solved in time~\smash{$2^{\order(n^{1/4})}$}, unless the Exponential Time Hypothesis fails.
\end{theorem}

A second question concerns the average case complexity of Sumset Recognition (and factoring). Can we generate hard instances by taking a random set $A$, and asking the algorithm to retrieve $A$ from the set $A+A$?

Finally, it would be interesting to study the approximability of the problem: given a set $S$, find the largest $S' \subseteq S$ such that $S'$ is a sumset.

\section{Preliminaries} \label{sec:preliminaries}
Throughout, we write $[n] = \set{1, \dots, n}$ and $[a,b] = \set{a,a+1, \dots,b}$ where $a<b$ are integers. We write $C = A \sqcup B$ whenever $C$ is the union of disjoint sets $A$ and $B$. For a subset $A \subseteq \Field_p^d$, and $v \in \Field_p^d$, we write $v \perp A$ if $v \mult u = 0$ for all $u \in A$.  
For subsets $A, B \subseteq G$ of some abelian group, we define their \emph{sumset} $A + B = \set{a + b : a \in A, b \in B}$. We write $r_{A + B}(x) = |\set{(a, b) \in A \times B : a + b = x}|$ to denote the number of representations of $x$ in the sumset $A + B$. Over finite fields, the following well-known fact about the size of $A + B$ will turn out useful; see~e.g.~\cite[Theorem~5.4]{TaoV06}.

\begin{lemma}[Cauchy-Davenport \cite{Cauchy1813,Davenport1935}] \label{lem:cauchy-davenport}
Let $A, B \subseteq \Field_p$. Then $|A + B| \geq \min(p, |A| + |B| - 1)$.
\end{lemma}

We also introduce some new terminology. Inspired by the analogy to factoring, we call a set $A$ \emph{irreducible} if it cannot be expressed as $A = B + C$ for sets of size $|B|, |C| \geq 2$. Let us call a set $A$ \emph{primitive} if, whenever we can express $A + A = B + B$, then $B = A + x$ for some group element $x$ of order $2$, i.e., $2x = 0$. That is, up to trivial shifts $x$, for a primitive set $A$ the representation as the sumset $A + A$ is unique.

\section{Positioning} \label{sec:pos}
In this section, we establish the first step in our chain of reductions and prove that the Sumset Recognition problem is equivalent to testing whether a set $X$ can be expressed as a sumset~\makebox{$X = A + A$} where $A \subseteq U$ is constrained by some prespecified superset $U$ (i.e., we ``position'' $A$). For the integers, we formally prove the following lemma:

\lemposint*

In \cref{sec:pos:sec:int}, we provide a proof of \cref{lem:pos-int}. This can also be considered as a ``warm-up'' for the analogous, considerably more involved positioning lemma over $\Field_p^d$ presented in \cref{sec:pos:sec:Fp}. Specifically, the advantage of the integers here is that we can solve many problems simply by choosing sufficiently large constants. This fails, for obvious reasons, over $\Field_p^d$.

\subsection{Positioning for \texorpdfstring{$\Int$}{Z}} \label{sec:pos:sec:int}
The idea behind our proof is that we first construct an appropriate ``skeleton set'' $I$ (see the discussion around \cref{lem:pos-int} in Section~\ref{sec:overview}). We will then identify certain length-$n$ intervals in $\Int$ with elements in the skeleton set, and construct~$X'$ in such a way that whenever $X' = A' + A'$, an interval in $A'$ is nonempty if and only if the corresponding element is present in the skeleton set. This forces $A'$ to have a lot of structure, which we will exploit to impose the superset condition. The precise structure we require is specified in the following lemma; while these conditions might appear obscure at first, we hope that they will become understandable in following the proof of \cref{lem:pos-int}.

\begin{lemma}[Skeleton Set for $\Int$] \label{lem:skeleton-int}
There are distinct elements $a, u, s \in [0, 2^{21}]$ and a set $I'' \subseteq [0, 2^{21}]$ satisfying the following conditions for $I' = \set{a, u, s - a, s - u}$:
\begin{enumerate}[label=(\roman*)]
    \item $I = I' \sqcup I''$ is primitive.
    \item $r_{I + I}(2a) = r_{I + I}(2s - 2a) = r_{I + I}(2s - 2u) = 1$ and $r_{I + I}(s) = 4$.
    \item $(I' + I') \setminus \set{2a, 2s - 2a, 2s - 2u, s} \subseteq I'' + I''$.
\end{enumerate}
\end{lemma}

In other words, the second item postulates that $2a$ has a unique representation in the sumset $I + I$, namely, $2a = a + a$. The same holds for $2s - 2a$ and $2s - 2u$. The element $s$ has four representations: $s = a + (s - a)$, $s = u + (s - u)$, and the two symmetric representations obtained by swapping the two summands; in the following, we will often neglect these symmetric options for simplicity.

\begin{proof}[Proof of \cref{lem:pos-int}]
For convenience, let us modify the input so that $\min(U) = \min(X) = 0$ and that $\max(U) = n$ and $\max(X) = 2n$ (by removing all elements larger than $\max(X) / 2$ from $U$ and shifting the sets if necessary; this transformation only decreases $n$). Let~\makebox{$a, u, s, I, I', I''$} be as in the previous lemma and write $m = 2^{21}$. We choose~\makebox{$X' = \set{i \cdot 4n + x : i \in I + I, x \in X'_i}$}, where the sets $X'_i \subseteq [0,  2n]$ are defined as follows:
\begin{alignat*}{3}
    &X'_s &&= U + \set{0, n}, \\
    &X'_{2a} &&= X, \\
    &X'_{2s - 2a} &&= X'_{2s - 2u} = \set{0, n, 2n}, \\
    &X'_{i} &&= [0,2n] &&\qquad(\text{for $i \in I + I''$}).
\end{alignat*}
\cref{lem:skeleton-int}~(iii) guarantees that these four cases indeed cover for all $i \in I + I$. This completes the description of $X'$; as claimed we have that $\max(X') \leq (4m + 2) n < 2^{24} n$.

\paragraph{Soundness}
Assume that there is a set $A \subseteq U$ with $X = A + A$; our goal is to show that there is some set $A' \subseteq \Int$ such that $X' = A' + A'$. To this end, we pick $A' = \set{i \cdot 4n + x : i \in I, x \in A'_i}$, where the sets $A'_i \subseteq [0,n]$ are defined as follows:
\begin{alignat*}{3}
    &A'_a &&= A, \\
    &A'_u &&= U, \\
    &A'_{s-a} &&= A'_{s-u} = \set{0, n}, \\
    &A'_i &&= [0,n] &&\qquad(\text{for $i \in I''$}).
\end{alignat*}

As the first step, we verify that $A' + A' \subseteq X'$. To this end, it suffices to check that~\makebox{$A'_i + A'_j \subseteq X'_{i+j}$} for all pairs~\makebox{$i, j \in I$}. This is obvious whenever $i + j \in I + I''$ since then~\smash{$X_{i+j}' = [0,2n]$}. By \cref{lem:skeleton-int}~(iii) the only other relevant cases are when $i + j \in \set{2a, 2s - 2a, 2s - 2u, s}$. By \cref{lem:skeleton-int}~(ii) the following cases are exhaustive: 
\begin{alignat*}{2}
    &A'_a + A'_a &&= A + A = X = X'_{2a}, \\
    &A'_{s - a} + A'_{s - a} &&= \set{0, n} + \set{0, n} = \set{0, n, 2n} = X'_{2s - 2a}, \\
    &A'_{s - u} + A'_{s - u} &&= \set{0, n} + \set{0, n} = \set{0, n, 2n} = X'_{2s - 2u}, \\
    &A'_a + A'_{s - a} &&= A + \set{0, n} \subseteq U + \set{0, n} = X'_s, \\
    &A'_u + A'_{s - u} &&= U + \set{0, n} = X'_s.
\end{alignat*}

Next, we verify that $X' \subseteq A' + A'$. To this end, we verify that for each $k \in I + I$ there exists some pair $i, j \in I$ such that $X'_k \subseteq A'_i + A'_j$. First, consider $k \in I + I''$. Observe that $\set{0, n} \subseteq A'_i$ for all $i \in I$. Recall that further $A'_j = \set{0, \dots, n}$ for all $j \in I''$, and hence,
\begin{equation*}
    X'_k \subseteq [0,2n] = \set{0, n} + [0,n].
\end{equation*}
By \cref{lem:skeleton-int}, these are the only remaining cases:
\begin{alignat*}{2}
    &X'_{2a} &&= X = A + A = A'_a + A'_a, \\
    &X'_{2s - 2a} &&= \set{0, n, 2n} = \set{0, n} + \set{0, n} = A'_{s - a} + A'_{s - a}, \\
    &X'_{2s - 2u} &&= \set{0, n, 2n} = \set{0, n} + \set{0, n} = A'_{s - u} + A'_{s - u}, \\
    &X'_s &&= U + \set{0, n} = A'_u + A'_{s-u}.
\end{alignat*}

\paragraph{Completeness}
Assume that there is a set $A' \subseteq \Int$ satisfying that $X' = A' + A'$. We construct a set $A \subseteq U$ with $X = A + A$. The first insight is that since $X' \subseteq \bigcup_{k \in I + I} [k \cdot 4n, k \cdot 4n + 2n]$, we have
\begin{equation*}
    A' \subseteq \bigcup_{0 \leq i \leq m} [i \cdot 2n, i \cdot 2n + n].
\end{equation*}
Importantly, $A'$ cannot contain elements from $[i \cdot 2n, i \cdot 2n + n]$ and $[j \cdot 2n, j \cdot 2n + n]$ at the same time when $i$ and $j$ have different parity. Since the smallest relevant $i$ is even (as otherwise we cannot cover $[k \cdot 4n,  k \cdot 4n + 2n]$ where $k = \min(I + I)$ is even), all relevant indices $i$ must be even. By rescaling all indices by $2$, we can thus express
\begin{equation*}
    A' = \bigcup_{0 \leq i \leq m} (4n \cdot i + A'_i),
\end{equation*}
for some sets $A'_i \subseteq [0,n]$. Moreover, we can rewrite the condition $X' = A' + A'$ as
\begin{equation*}
    X'_k = \bigcup_{\substack{0 \leq i, j \leq m\\i + j = k}} (A'_i + A'_j).
\end{equation*}
Thus, letting $J = \set{i : A'_i \neq \emptyset}$ we must have that $I + I = J + J$. Since the set $I$ is primitive by \cref{lem:skeleton-int}~(i), we conclude that $I = J$. Next, recall that $r_{I + I}(2a) = 1$ by \cref{lem:skeleton-int}~(ii), and thus
\begin{equation*}
    X = X'_{2a} = A'_a + A'_a.
\end{equation*}
\cref{lem:skeleton-int}~(ii) further guarantees that $2s - 2a$ has a unique representation, and so
\begin{alignat*}{2}
    \set{0, n, 2n} &= X'_{2s - 2a} &&= A'_{s - a} + A'_{s - a},
\end{alignat*}
entailing that $A'_{s - a} = \set{0, n}$. We finally use that
\begin{alignat*}{2}
    A'_a + \set{0, n} &= A'_a + A'_{s - a} &&\subseteq X'_s = U + \set{0, n},
\end{alignat*}
which implies that $A'_a \subseteq U$. In summary, the set $A := A'_a$ satisfies both conditions that $X = A + A$ and that $A \subseteq U$.

\paragraph{Running Time}
Let us finally comment on the running time. The set $I$ as provided by \cref{lem:skeleton-int} has constant size, and its construction time can thus be neglected. Given this, the set $X$ can easily be constructed in polynomial time by appropriately shifting and copying $X$ and $U$.
\end{proof}

It remains to prove \cref{lem:skeleton-int}. This involves, in particular, the construction of a primitive set, and for this purpose, the following lemma will turn out useful:

\begin{lemma} \label{lem:make-primitive-int}
For any set $A \subseteq [0,n]$, the set $A \cup \set{4n}$ is primitive.
\end{lemma}
\begin{proof}
Write $A^* = A \cup \set{4 n}$, and suppose that $A^* + A^* = B + B$. We show that $B = A^*$. First, observe that since
\begin{equation*}
    A^* + A^* = (A + A) \cup (A + 4n) \cup \set{8n} \subseteq [0,2n] \cup [4n,5n] \cup \set{8n},
\end{equation*}
we can only have that $B \subseteq [0,n] \cup [2n, \floor{\frac52 n}] \cup \set{4n}$. Clearly, $4n \in B$ as otherwise, we cannot generate $8n \in B + B$. Hence, $B$ cannot contain any element in $[2n, \floor{\frac52 n}]$ as otherwise $B + B$ would contain an element in $[6n, \floor{\frac{13}2 n}]$. In summary, we can write $B = B_0 \cup \set{4n}$ for some set $B_0 \subseteq [0,n]$. Finally, expand $B + B = (B_0 + B_0) \cup (B_0 + 4n) \cup \set{8n}$ and observe that the middle term implies that $B_0 = A$ and thus $B = A^*$.
\end{proof}

\begin{proof}[Proof of \cref{lem:skeleton-int}]
We choose $(j_1, \dots, j_5, a, u, s) := (5^0, 5^1, \dots, 5^7)$; our intention behind this choice is that any multi-subset of $\set{j_1, \dots, j_5, a, u, s}$ with multiplicity at most $4$ sums to a unique value. Let $E = (I' + I') \setminus \set{2a, 2s - 2a, 2s - 2u, s}$. Then, noting that $|I' + I'| = 9$ (as all pairs of elements in the size-4 set $I'$ sum up to unique values  
except for $s$ that can be expressed as $(s-u) + u$ and $(s-a) + a$), it follows that~\smash{$|E| = |I'+I'| - 4 = 9 - 4 = 5$}. By arbitrarily identifying~$E$ and~$[1,5]$ we choose
\begin{equation*}
    I'' = \set{e - j_e, j_e : e \in E} \cup \set{4 \cdot 5^8}.
\end{equation*}
It is easy to verify that $\max(I) = 4 \cdot 5^8 \leq 2^{21}$. In the following, we argue that the properties~(i),~(ii), and~(iii) hold.
\begin{enumerate}[label=(\roman*)]
    \item Note that $I = I' \sqcup I''$ can be written as $A \sqcup \set{4 \cdot 5^8}$ where $A \subseteq [0,5^8]$. By \cref{lem:make-primitive-int}, $I$ is thus indeed primitive.
    \item The goal is to show that $2a$, $2s - 2a$, $2s - 2u$ can be uniquely expressed in $I + I$, and that $s = a + (s - a) = u + (s - u)$ has exactly these two representations (up to exchanging the summands). This can easily be checked using that any multi-subset of $I$ with multiplicity at most $4$ sums to a unique value. (In particular, $I''$ cannot participate.)
    \item This property equivalently states that $E \subseteq I'' + I''$. And indeed, each $e \in E$ can be expressed as $e = (e - j_e) + j_e \in I''$. \qedhere
\end{enumerate}
\end{proof}

We remark that we made no attempt to optimize the constants in \cref{lem:skeleton-int}, but rather aimed for a proof that is as simple as possible.

\subsection{Positioning for \texorpdfstring{$\Field_p^d$}{Fpd}} \label{sec:pos:sec:Fp}
In this section, we will rework the entire positioning proof and present its adaption to $\Field_p^d$. We encourage the reader to skip this section at first reading and to continue reading the overall reduction in \cref{sec:masking,sec:sat}.

The additional difficulty for $\Field_p^d$ is mostly due to the cases $\Field_2^d$ and $\Field_3^d$. In both cases, it becomes somewhat more involved to construct primitive sets (see the following two lemmas). However, the more bothersome difficulty is that over $\Field_2^d$ we cannot easily require that $X = A + A$ with the same proof strategy. Specifically, in the language of \cref{lem:pos-int} the issue is that if we want to place~$A$ at some position $A'_a$, then $A + A$ is placed at~$A'_0$. However, all sorts of garbage terms are also placed at~$A'_0$, making the constraint void. We deal with this issue by designing a weaker positioning reduction involving two sets $A$ and $B$; see \cref{lem:pos-Fp}.

\begin{lemma} \label{lem:make-irreducible}
Let $A \subseteq \Field_p^d$ with $|A| \geq 2$, and let $v \perp A$. Then $A \cup \set{v}$ is irreducible.
\end{lemma}
\begin{proof}
Suppose that $A \cup \set{v} = B + C$; we show that $|B| \leq 1$ or $|C| \leq 1$. Let $V \subseteq \Field_p^d$ denote the subspace orthogonal to $v$, and  for $i \in \Field_p$ let
\begin{align*}
    B_i &= \set{x \in V : x + i v \in B}, \\
    C_i &= \set{x \in V : x + i v \in C}.
\end{align*}
Note that $A \cup \set{v} = B + C$ can be rewritten in terms of the following three conditions:
\begin{align*}
    A &= \bigcup_{\substack{i, j \in \Field_p\\i + j = 0}} (B_i + C_j), \\
    \set{0} &= \bigcup_{\substack{i, j \in \Field_p\\i + j = 1}} (B_i + C_j), \\
    \emptyset &= \bigcup_{\substack{i, j \in \Field_p\\i + j \neq 0, 1}} (B_i + C_j).
\end{align*}
Note further that $|B| = \sum_i |B_i|$ and that $|C| = \sum_i |C_i|$.

We first establish that $|B_i|, |C_i| \leq 1$. Suppose for contradiction that $|B_i| \geq 2$. Then by the second and third conditions, we have $|C_{j}| = 0$ for all $j \neq -i$. If $|C_{-i}| \leq 1$ then also $|C| \leq 1$ and we are done. Instead, assume that $|C_{-i}| \geq 2$. Then by the symmetric argument, we have that $|B_j| = 0$ for all $j \neq i$. But this violates the second condition as $\bigcup_{i + j = 1} (B_i + C_j) = \emptyset$.

To complete the proof, let $I = \set{i : B_i \neq \emptyset}$ and $J = \set{i : C_i \neq \emptyset}$. In order to satisfy the three previous conditions, we must have that $I + J = \set{0, 1}$. The Cauchy-Davenport theorem implies that either $|I| \leq 1$ or $|J| \leq 1$ or $p = 2$. In the former two cases we are done, as then~\smash{$|B| = \sum_{i \in I} |B_i| \leq 1$} or~\smash{$|C| = \sum_{i \in J} |C_j| \leq 1$}. In the latter case, if $p = 2$, there are elements $x, y$ such that $B_0, C_1 \subseteq \set{x}$ and $B_1, C_0 \subseteq \set{y}$ in order to satisfy the second condition. But this implies that $A \subseteq \set{x + y}$, contradicting the assumption that $|A| \geq 2$.
\end{proof}

\begin{lemma} \label{lem:make-primitive}
Let $A \subseteq \Field_p^d$ be irreducible with $|A| \geq 2$ and let $v \perp A$. Then $A \cup \set{v}$ is primitive.
\end{lemma}
\begin{proof}
Write $A^* = A \cup \set{v}$ and suppose that $B + B = A^* + A^* = (A + A) \cup (v + A) \cup \set{2v}$; we show that $B = A^* + t$ for some shift $t$ with $2t = 0$. Let $V$ denote the $(d-1)$-dimensional subspace orthogonal to $v$, and for $i \in \Field_p$ write $B_i = \set{x \in V : x + i v \in B}$. 
Using this notation, observe that the assumption $B + B = A^* + A^*$ can be rewritten as
\begin{alignat*}{2}
    A + A &= \bigcup_{\substack{i, j \in \Field_p\\i + j = 0}} (B_i + B_j), \\
    A &= \bigcup_{\substack{i, j \in \Field_p\\i + j = 1}} (B_i + B_j), \\
    \set{0} &= \bigcup_{\substack{i, j \in \Field_p\\i + j = 2}} (B_i + B_j) &&\qquad(\text{only relevant if $p > 2$}), \\
    \emptyset &= \bigcup_{\substack{i, j \in \Field_p\\i + j \neq 0, 1, 2}} (B_i + B_j) &&\qquad(\text{only relevant if $p > 3$}).
\end{alignat*}
We distinguish three cases for $p$:

\begin{itemize}
    \item $p = 2$: In this case the second identity states that $A = B_0 + B_1$. By the irreducibility of $A$, there is some shift $s$ such that either $B_0 = A + s$ and $B_1 = \set{s}$, or~\makebox{$B_0 = \set{0}$} and $B_1 = A + s$. Let $t = s$ in the former case and $t = s + v$ in the latter case. Then, since $B = B_0 \cup (v + B_1)$, we indeed have $B = A^* + t$.
    \item $p = 3$: The third condition implies that $B_1 + B_1 \subseteq \set{0}$ and hence $B_1 \subseteq \set{0}$. It also implies that $B_0 + B_2 \subseteq \set{0}$, which can only be satisfied in the following three cases:
    \begin{itemize}
        \item $|B_0| = |B_2| = 1$: In this case the first identity implies that $|A + A| \leq 2$. But recall that $|A| \geq 2$, and thus $|A| = |A + A| = 2$. This can only happen when $A$ is a $1$-dimensional subspace of $\Field_2^d$. But then $A = A + A$, which contradicts the irreducibility of $A$.
        \item $|B_0| = 0$: Then the second identity implies that $A = B_2 + B_2$. This again contradicts the irreducibility of $A$.
        \item $|B_2| = 0$: Then the second inequality implies that $B_0 = A$ and that $B_1 = \set{0}$. Putting these together, we find that $B = A^*$ as claimed.
    \end{itemize}
    \item $p > 3$: Let $I = \set{i : B_i \neq \emptyset}$ and note that $I + I = \set{0, 1, 2}$. The Cauchy-Davenport theorem implies that $|I| \leq 2$---in fact, the only feasible solution is $I = \set{0, 1}$.\footnote{To see this, first observe that $I$ can only consist of $0$, $1$ and $2^{-1}$. Setting $I = \set{0, 2^{-1}}$ fails as we need that $0 + 2^{-1} = 2$ (which is only satisfied for $p = 3$). Similarly, setting $I = \set{1, 2^{-1}}$ fails as we require that $1 + 2^{-1} = 0$ (which again is only satisfied for $p = 3$).} Therefore, and as~\makebox{$B_1 \subseteq \set{0}$} by the third identity, the second identity implies that $B_0 = A$ and $B_1 = \set{0}$. This shows that $B = A^*$ as claimed. \qedhere
\end{itemize}
\end{proof}

Putting these two lemmas together, we can make an arbitrary set $A$ primitive by including two vectors $v_1, v_2$ satisfying the orthogonality constraints. Conversely, we remark that there are indeed examples where \emph{two} orthogonal vectors $v_1, v_2$ are necessary to make $A$ primitive. For instance, take a $(d-1)$-dimensional subspace $V \subseteq \Field_3^d$ orthogonal to some nonzero vector $v$. Then the set $A \cup {v}$ is not primitive as $(A \cup \set{v}) + (A \cup \set{v}) = B + B$ where $B = (A - v) \cup \set{v}$.

\begin{lemma}[Skeleton Set for $\Field_p^d$] \label{lem:skeleton-Fp}
For any prime $p$, there are distinct elements $a, b, u, v, s, t \in \Field_p^{40}$ and a set $I'' \subseteq \Field_p^{40}$ satisfying the following conditions for $I' = \set{a, b, u, v, s - a, s - u, t - b, t - v}$:
\begin{enumerate}[label=(\roman*)]
    \item $I = I' \sqcup I''$ is primitive.
    \item $r_{I + I}(a + b) = r_{I + I}(s - a + t - u) = 2$ and $r_{I + I}(s) = r_{I + I}(t) = 4$.
    \item $(I' + I') \setminus \set{a + b, s - a + t - u, s, t} \subseteq I'' + I''$.
\end{enumerate}
\end{lemma}
Put in words, the second condition guarantees that $a + b$ and $(s - a) + (t - u)$ can be uniquely represented in the sumset $I + I$ (up to exchanging the summands), and that $s$ and $t$ have exactly two representations (namely, $s = a + (s - a)$ and $s = u + (s - u)$, and similarly for $t$).

\begin{proof}
Let $a, b, u, v, s, t, j_1, \dots, j_{32}, v_1, v_2 \in \Field_p^{40}$ denote pairwise orthogonal nonzero vectors, and write $E = (I' + I') \setminus \set{a + b, s - a + t - u, s, t}$; note that~\smash{$|E| \leq \binom{|I'|+1}{2} - 4 = \binom{9}{2} - 4 = 32$}. By arbitrarily identifying $E$ with $[1,32]$, we let
\begin{equation*}
    I'' = \set{e + j_e, -j_e : e \in E} \sqcup \set{v_1, v_2}.
\end{equation*}
In the following, we argue that the properties (i), (ii) and (iii) hold.

\begin{enumerate}[label=(\roman*)]
    \item Note that we can write $I = A \sqcup \set{v_1, v_2}$ where clearly $|A| \geq 2$ and $A \perp v_1, v_2$ and $v_1 \perp v_2$. Thus, by \cref{lem:make-irreducible} the set $A \sqcup \set{v_1}$ is irreducible, and by \cref{lem:make-primitive} the set $A \sqcup \set{v_1, v_2}$ is primitive.
    \item Condition (ii) enforces specific multiplicities for exactly the elements in $x \in (I' + I') \setminus E$. For any such element $x$, it can easily be verified that $x \not\in I'' + I'$ (since all elements in $I''$ have a nontrivial component $j_e$, but $j_e \perp I'$) and that $x \not\in I'' + I''$ (by construction we can only express the elements in $E$ in this way). It follows that $r_{I + I}(x) = r_{I' + I'}(x)$. Having this in mind, we can verify condition (ii) by hand: The sums $a + b$ and $(s - a) + (t - b)$ are unique, $s$ appears exactly two times as $s = a + (s - a) = u + (s - u)$ and similarly $t$ appears exactly two times as $t = b + (t - b) = v + (t - v)$.
    \item This condition is trivial by construction: We can express $e = e + j_e - j_e \in I'' + I''$ for all elements $e \in E$. \qedhere
\end{enumerate}
\end{proof}

\begin{lemma}[Positioning for $\Field_p^d$] \label{lem:pos-Fp}
Let $p$ be a prime and let $X, U, V \subseteq \Field_p^d$. In time $\poly(p^d)$ we can construct a set~\smash{$X' \subseteq \Field_p^{d'}$}, where $d' = d + 40$, such that:
\begin{equation*}
    \exists A \subseteq U, B \subseteq V : X = A + B \qquad\text{if and only if}\qquad \exists A' \subseteq \Field_p^{d'} : X' = A' + A'.
\end{equation*}
\end{lemma}
\begin{proof}
We start with the construction of $X'$. Let $a, b, u, v, s, t, I, I', I''$ be as in the previous lemma. We view~\smash{$\Field_p^{d'} = \Field_p^{40} \times \Field_p^d$} and set $X' = \set{(i, x) : i \in I + I, x \in X'_i}$, where the sets $X'_i \subseteq \Field_p^d$ are defined as follows:
\begin{alignat*}{3}
    &X'_{a + b} &&= X, \\
    &X'_{s} &&= U, \\
    &X'_{t} &&= V, \\
    &X'_{s - a + t - b} &&= \set{0}, \\
    &X'_{i} &&= \smash{\Field_p^d} &&\qquad(\text{for $i \in I + I''$}).
\end{alignat*}
\cref{lem:skeleton-Fp}~(iii) guarantees that these four cases indeed cover for all $i \in I + I$.

\paragraph{Soundness}
Assume that there are sets $A \subseteq U$ and $B \subseteq V$ with $X = A + B$. We show that there is some set $A' \subseteq \Field_p^{d'}$ such that $X' = A' + A'$. We choose $A' = \set{(i, x) : i \in I, x \in A'_i}$, where the sets~$A'_i$ are defined as follows:
\begin{alignat*}{3}
    &A'_a &&= A, \\
    &A'_b &&= B, \\
    &A'_u &&= U, \\
    &A'_v &&= V, \\
    &A'_{s - a} &&= A'_{s - u} = A'_{t - b} = A'_{t - v} = \set{0}, \\
    &A'_i &&= \smash{\Field_p^d} &&\qquad(\text{for $i \in I''$}).
\end{alignat*}

We first verify that $A' + A' \subseteq X'$. To this end, we check that $A'_i + A'_j \subseteq X'_{i+j}$ for all pairs~\makebox{$i, j \in I$}. This is clearly true whenever $i + j \in I + I''$ since~\smash{$X_{i+j}' = \Field_p^d$}. By \cref{lem:skeleton-Fp}~(iii), the only other relevant cases are when $i + j \in \set{a + b, s - a + t - b, s, t}$. By \cref{lem:skeleton-Fp}~(ii) the following cases are exhaustive: 
\begin{alignat*}{3}
    &A'_a + A'_b &&= A + B &&= X = X'_{a + b}, \\
    &A'_{s - a} + A'_{t - b} &&= \set{0} + \set{0} &&= \set{0} = X'_{s - a + t - b}, \\
    &A'_a + A'_{s - a} &&= A + \set{0} &&= A \subseteq U = X'_s, \\
    &A'_u + A'_{s - u} &&= U + \set{0} &&= U = X'_s, \\
    &A'_b + A'_{t - b} &&= B + \set{0} &&= B \subseteq V = X'_t, \\
    &A'_v + A'_{t - v} &&= V + \set{0} &&= V = X'_t.
\end{alignat*}

Next, we verify that $X' \subseteq A' + A'$. To this end, we verify that for each $k \in I + I$ there exists some pair $i, j \in I$ such that $X'_k \subseteq A'_i + A'_j$. For each $k \in I + I''$ this is clear using that~\smash{$A'_i = \Field_p^d$} for all $i \in I''$. By \cref{lem:skeleton-Fp}, these are the remaining cases:
\begin{alignat*}{4}
    &X'_{a + b} &&= X &&= A + B &&= A'_a + A'_b, \\
    &X'_s &&= U &&= U + \set{0} &&= A'_u + A'_{s - u}, \\
    &X'_t &&= V &&= V + \set{0} &&= A'_v + A'_{t - v}, \\
    &X'_{s - a + t - b} &&= \set{0} &&= \set{0} + \set{0} &&= A'_{s - a} + A'_{t - b}.
\end{alignat*}

\paragraph{Completeness}
Next, assume that there is a set $A' \subseteq \Field_p^{d'}$ such that $X' = A' + A'$. Our goal is to construct sets $A \subseteq U$ and $B \subseteq V$ with $X = A + B$. For~\smash{$i \in \Field_p^{40}$}, let us write~\smash{$A'_i = \set{x \in \Field_p^d : (i, x) \in A'}$}, and let $J = \set{i \in \Field_p^{40} : A'_i \neq \emptyset}$. As we assume that $X' = A' + A'$, we must have that $I + I = J + J$. Since $I$ is primitive by \cref{lem:skeleton-Fp}~(i), we conclude that $J = I + s$ for some shift $s$ with $2s = 0$; without loss of generality we assume that $s = 0$. Since $a + b$ has a unique representation in $I + I$ by \cref{lem:skeleton-Fp}~(ii), we have that
\begin{equation*}
    X = X'_{a + b} = A'_a + A'_b.
\end{equation*}
Moreover, by \cref{lem:skeleton-Fp}~(ii), $(s - a) + (t - b)$ also has a unique representation, and therefore
\begin{alignat*}{2}
    \set{0} &= X'_{s - a + t - b} &&= A'_{s - a} + A'_{t - b}.
\end{alignat*}
It follows that for some element $x \in \Field_p^d$ we have $A'_{s - a} = \set{x}$ and $A'_{t - b} = \set{-x}$. We finally have that
\begin{alignat*}{2}
    A'_a + x &= A'_a + A'_{s - a} &&\subseteq X'_s = U, \\
    A'_b - x &= A'_b + A'_{t - b} &&\subseteq X'_t = V.
\end{alignat*}
In summary, the sets $A := A'_a + x$ and $B := A'_b - x$ satisfy the three conditions that $X = A + B$ and $A \subseteq U$ and $B \subseteq V$.
\end{proof}

\section{Masking} \label{sec:masking}
In this section, we implement the ``masking'' step, which transforms the problem of testing the exact sumset condition $X = A + A$ to a more relaxed version testing whether $X \subseteq A + A \subseteq Y$. Specifically, for the integers, we prove the following reduction:

\lemmaskingint*

In \cref{sec:masking:sec:int}, we provide a proof of \cref{lem:masking-int}, and in the following \cref{sec:masking:sec:Fp}, we design the appropriate analogue for finite field vector spaces.

\subsection{Masking for \texorpdfstring{$\Int$}{Z}} \label{sec:masking:sec:int}
Unfortunately, the proof of \cref{lem:masking-int} turns out to be quite technical. While in the positioning step working with integers had several advantages, in the masking step, we have to deal with several drawbacks. One particular drawback is as follows: In order to ``mask'' a set $A$, we want to find another set $M$ such that $A + M$ is predictable without knowing $A$. Over~\smash{$\Field_p^d$} this easily works by choosing~\smash{$M = \Field_p^d$}, in which case~\smash{$A + M = \Field_p^d$} without having any knowledge about $A$ (other than the trivial condition that $A$ be nonempty). Over the integers, the natural analogue is to pick an interval $M = [0,N]$. However, then $A + M = [\min(A), \max(A) + N]$, even if we presuppose that~\makebox{$A \subseteq [0,N]$}, which is inconveniently not independent of $A$. Throughout the proof, this issue will introduce several less important side cases. The first step in dealing with it is the following lemma:

\begin{lemma} \label{lem:range-int}
Let $X, Y, U \subseteq [0,n]$. We can compute sets $X', Y', U' \subseteq [0,32n]$ in time $\poly(n)$ such that:
\begin{enumerate}[label=(\roman*)]
    \item $\exists A \subseteq U : X \subseteq A + A \subseteq Y \quad\text{if and only if}\quad \exists A' \subseteq U' : X' \subseteq A' + A' \subseteq Y'$.
    \item All elements in $X', Y', U'$ are even.
    \item $\min(X') = \min(Y') = \min(U') = 0$, $\max(X') = \max(Y') = 32n$, $\max(U') = 16n$.
\end{enumerate}
\end{lemma}
\begin{proof}
We write $2X = \set{2x : x \in X}$. The construction is simple (emphasizing again that our goal is not to optimize the constants but rather to achieve simple proofs):
\begin{align*}
    X' &= \set{0} \cup (2X + 8n) \cup \set{32n}, \\
    Y' &= 2[0,3n] \cup (2Y + 8n) \cup 2[8n,16n], \\
    U' &= \set{0} \cup (2U + 4n) \cup \set{16n}.
\end{align*}
For the soundness, let $A \subseteq U$ satisfy that $X \subseteq A + A \subseteq Y$. We choose $A' = \set{0} \cup (2A + 4n) \cup \set{16n}$.
Noting that
\begin{equation*}
    A' + A' = \set{0} \cup (2A + 4n) \cup (2A + 2A + 8n) \cup \set{16n} \cup (2A + 20n) \cup \set{32n},
\end{equation*}
it is straightforward to verify that $A' \subseteq U'$ and that $X' \subseteq A' + A' \subseteq Y'$.

For the converse direction, suppose that $A' \subseteq U'$ satisfies that $X' \subseteq A' + A' \subseteq Y'$. Let $A$ be the part of $A'$ that falls in the interval $[4n,6n]$ where we divide each number by $2$ (here we use that $A'$ contains only even numbers since $A' \subseteq U'$). Using again that $A' \subseteq U'$ we can express~$A'$ as~$\set{0} \cup (A + 4n) \cup \set{16n}$. It follows easily that $X \subseteq A + A \subseteq Y$.
\end{proof}

\begin{proof}[Proof of \cref{lem:masking-int}]
By first applying the transformation from \cref{lem:range-int}, we assume throughout that $X, Y, U \subseteq [0,32n]$ with $\min(X) = \min(Y) = \min(U) = 0$, $\max(X) = \max(Y) = 32n$ and $\max(U) = 16n$, and that all three sets contain only even numbers. For simplicity, we write~\makebox{$N = 16n$}. Furthermore, if $X \not\subseteq Y$, then the reduction is trivial, so assume that $X \subseteq Y$ and let $Z = Y \setminus X$. In this proof, we will construct sets~\makebox{$X', U' \subseteq \Int$} containing negative integers; this can later be fixed by shifting all sets appropriately into the positive range.

The construction of the sets $X'$ and $U'$ is, unfortunately, extremely tedious. We define
\begin{align*}
    X' = \bigcup_{\substack{-2 \leq c \leq 2\\-8N \leq k \leq 8N}} (c \cdot 64N^2 + k \cdot 4N + X'_{c, k}), \\
    U' = \bigcup_{\substack{-1 \leq a \leq 1\\-4N \leq i \leq 4N}} (a \cdot 64N^2 + i \cdot 4N + U'_{a, i}),
\end{align*}
where the pieces $X'_{c, k} \subseteq [0,2N]$ and $U'_{a, i} \subseteq [0,N]$ are defined as follows; we write $I = [-4N,4N] \setminus \set{0}$ for simplicity here:
\begin{alignat*}{3}
    &X'_{2, k} &&= [0,2N] &&\qquad(\text{for $k \in I + I$}), \\
    &X'_{1, k} &&= [0,2N] &&\qquad(\text{for $k \in I$}), \\
    &X'_{1, k} &&= \bigcup_{\substack{i \in I, z \in Z\\i \pm z = k}} [z/2, z/2 + N] &&\qquad(\text{for $k \not\in I$}), \\
    &X'_{0, 0} &&= Y, \\
    &X'_{0, k} &&= [0,2N] &&\qquad(\text{for $k \in I$}), \\
    &X'_{-1, 0} &&= [0,2N], \\
    &X'_{-1, z} &&= X'_{-1, -z} = [z/2, z/2 + N] &&\qquad(\text{for $z \in Z$}), \\
    &X'_{-2, 0} &&= [0,2N],
\end{alignat*}
and
\begin{alignat*}{3}
    &U'_{1, i} &&= [0,N] &&\qquad(\text{for $i \in I$}), \\
    &U'_{0, 0} &&= U, \\
    &U'_{0, z} &&= U'_{0, -z} = \set{z/2} &&\qquad(\text{for $z \in Z$}), \\
    &U'_{-1, 0} &&= [0,N].
\end{alignat*}
Moreover, all sets $X'_{\ell, i}$ and $U'_{\ell, i}$ which are not covered by these cases are defined to be empty.

Before checking the soundness and completeness properties, we first check that the size of~$X', U'$ is appropriate. We have that $X', U' \subseteq [-256N^2, 256N^2]$, and thus by shifting them into the positive range, we would obtain sets in $[0, 512N^2]$. Recalling that $N = 16n$, the universe bound $512 N^2 = 2^{17} n$ holds.

\paragraph{Soundness}
Suppose that there is some set $A \subseteq U$ with $X \subseteq A + A \subseteq Y$. We give a set $A' \subseteq U'$ with $X' = A' + A'$. As before, we define
\begin{equation*}
    A' = \bigcup_{\substack{-1 \leq a \leq 1\\-4N \leq i \leq 4N}} (a \cdot 64N^2 + i \cdot 4N + A'_{a, i})
\end{equation*}
where $A'_{a, i} \subseteq [0,n]$ is defined as follows:
\begin{alignat*}{3}
    &A'_{1, i} &&= [0,N] &&\qquad(\text{for $i \in I$}), \\
    &A'_{0, 0} &&= A, \\
    &A'_{0, z} &&= A'_{0, -z} = \set{z/2} &&\qquad(\text{for $z \in Z$}), \\
    &A'_{-1, 0} &&= [0,N].
\end{alignat*}
From the construction, it is immediate that $A' \subseteq U'$. Observe further that $0, N \in A$ (as otherwise we cannot satisfy that $\set{0, 2N} \subseteq X \subseteq A + A$), and thus 
$A + [0,N] = [0,2N]$; this observation will be useful in the following paragraphs.

Next, we show that $A' + A' \subseteq X'$. To this end, we verify that $A'_{a, i} + A'_{b, j} \subseteq X'_{a+b,i+j}$ for all choices~$a, b, i, j$. This is not hard, but it requires a lot of case work. We state the relevant cases and omit some others due to symmetries:
\begin{alignat*}{3}
    &A'_{1, i} + A'_{1, j} &&= [0,2N] = X'_{2, i+j} &&\qquad\text{(for $i, j \in I$)}, \\
    &A'_{1, i} + A'_{0, 0} &&= [0,N] + A = [0,2N] = X'_{1, i} &&\qquad\text{(for $i \in I$)}, \\
    &A'_{1, i} + A'_{0, z} &&= [0,N] + \set{z/2} = [z/2,  z/2 + N] \subseteq X'_{1, i + z} &&\qquad\text{(for $i \in I, z \in Z$)}, \\
    &A'_{1, i} + A'_{-1, 0} &&= [0,2N] = X'_{0, i} &&\qquad\text{(for $i \in I, z \in Z$)}, \\
    &A'_{0, 0} + A'_{0, 0} &&= A + A \subseteq Y = X'_{0, 0}, \\
    &A'_{0, 0} + A'_{0, z} &&= A + \set{z/2} \subseteq [0,2N] = X'_{0, z} &&\qquad\text{(for $z \in Z$)}, \\
    &A'_{0, 0} + A'_{-1, 0} &&= A + [0,N] = [0,2N] = X'_{-1, 0}, \\
    &A'_{0, z} + A'_{0, z'} &&= \set{z/2} + \set{z/2} \subseteq [0,2N] = X'_{0, z+z'} &&\qquad\text{(for $z, z' \in Z$)}, \\
    &A'_{0, z} + A'_{0, -z} &&= \set{z/2} + \set{z/2} = \set{z} \subseteq Z \subseteq Y = X'_{0, 0} &&\qquad\text{(for $z \in Z$)}, \\
    &A'_{0, z} + A'_{-1, 0} &&= \set{z/2} + [0,N] = [z/2, z/2 + N] = X'_{-1, z} &&\qquad\text{(for $z \in Z$)}, \\
    &A'_{-1, 0} + A'_{-1, 0} &&= [0, 2N] = X'_{-2, 0}.
\end{alignat*}

Let us finally check that $X' \subseteq A' + A'$. That is, our task is to verify that for all pairs $c, k$ it holds that
\begin{equation*}
    X'_{c, k} \subseteq \bigcup_{\substack{a + b = c\\i + j = k}} (A'_{a, i} + A'_{b, j}).
\end{equation*}
This again involves dealing with many cases:
\begin{alignat*}{3}
    &X'_{2, k} &&= [0,2N] = A'_{1, i} + A'_{1, j} &&\qquad\text{(for $k = i + j \in I + I$)}, \\
    &X'_{1, k} &&= [0,2N] = A + [0,N] = A'_{0, 0} + A'_{1, k} &&\qquad\text{(for $k \in I$)}, \\
    &X'_{1, k} &&= \bigcup_{\substack{i \in I, z \in Z\\i \pm z = k}} [z/2, z/2 + N] = \bigcup_{\substack{i \in I, z \in Z\\i \pm z = k}} (A'_{1, i} + A'_{0, z}) &&\qquad\text{(for $k \not\in I$)}, \\
    &X'_{0, 0} &&= Y = X \cup Z \\
    & &&\subseteq (A + A) \cup \bigcup_{z \in Z} (\set{z/2} + \set{z/2}) \\
    & &&= (A'_{0, 0} + A'_{0, 0}) \cup \bigcup_{z \in Z} (A'_{0, z} \cup A'_{0, -z}), \\
    &X'_{0, k} &&= [0,2N] = A'_{1, k} + A'_{-1, 0} &&\qquad\text{(for $k \in I$)}, \\
    &X'_{-1, 0} &&= [0, 2N] = A + [0,N] = A'_{0, 0} + A'_{-1, 0}, \\
    &X'_{-1, z} &&= [z/2, z/2 + N] = [0,N] + \set{z/2} = A'_{-1, 0} + A'_{0, z} &&\qquad\text{(for $z \in Z$)}, \\
    &X'_{-2, 0} &&= [0, 2N] = A'_{-1, 0} + A'_{-1, 0}.
\end{alignat*}

\paragraph{Completeness}
The completeness proof is simpler in comparison. Suppose that there is a set~\makebox{$A' \subseteq U'$} with $X' = A' + A'$; we show that there is some set $A \subseteq U$ satisfying that $X \subseteq A + A \subseteq Y$. Let, in analogy to before, $A'_{a, i} \subseteq [0,N]$ denote the subset of $A'$ that falls into the length-$N$ range starting at~\smash{$a \cdot 64N^2 + i \cdot 4N$}. These are the only nonempty regions in $A'$ (by the condition that~\makebox{$A' \subseteq U'$}), and therefore it holds that
\begin{equation*}
    X'_{c, k} = \bigcup_{\substack{a + b = c\\i + j = k}} (A'_{a, i} + A'_{b, j}).
\end{equation*}
On the one hand, it is clear that $A'_{0, 0} + A'_{0, 0} \subseteq X'_{0, 0} = Y$. On the other hand, by the construction, it is easy to verify that each element $x \in Y \setminus Z$ can only appear in 
\begin{equation*}
    \bigcup_{\substack{a + b = 0\\i + j = 0}} (A'_{a, i} + A'_{b, j}) = (A'_{0, 0} + A'_{0, 0}) \cup \bigcup_{z \in Z} (A'_{0, z} + A'_{0, -z})
\end{equation*}
if $x \in A'_{0, 0} + A'_{0, 0}$. In particular, this implies that $X = Y \setminus Z \subseteq A'_{0, 0} + A'_{0, 0}$. Finally, observe that trivially $A'_{0, 0} \subseteq U'_{0, 0} = U$, and therefore $A'_{0, 0}$ satisfies all three relevant conditions.

\paragraph{Running Time}
The running time due to \cref{lem:range-int} is negligible, and we can easily compute the sets $X', U'$ in polynomial time.
\end{proof}

\subsection{Masking for \texorpdfstring{$\Field_p^d$}{Fpd}} \label{sec:masking:sec:Fp}
In comparison to the previous subsection, the masking proof over $\Field_p^d$ is surprisingly simple.

\begin{lemma}[Masking for $\Field_p^d$] \label{lem:masking-Fp}
Let $p$ be a prime and let $X, Y, U, V \subseteq \Field_p^d$. In time~$\poly(p^d)$ we can construct sets~\smash{$X', U', V' \subseteq \Field_p^{d'}$}, where $d' = 2d + 3$, such that:
\begin{equation*}
    \exists A \subseteq U, B \subseteq V : X \subseteq A + B \subseteq Y \qquad\text{if and only if}\qquad \exists A' \subseteq U', B' \subseteq V' : X' = A' + B'.
\end{equation*}
\end{lemma}
\begin{proof}
If $X \not\subseteq Y$ then the reduction is trivial, so assume that $X \subseteq Y$ and let $Z = Y \setminus X$. Consider the vector space~\smash{$\Field_p^{d+3}$}; by embedding $\Field_p^d$ into the first $d$ coordinates, say, we view $Z$ also as a subset of~\smash{$\Field_p^{d+3}$}. We exploit the remaining three dimensions by letting $a, b, w \perp Z$ denote nonzero pairwise orthogonal vectors; let $W \subseteq \Field_p^d$ denote the subspace orthogonal to $w$. We define~\smash{$X' = \set{(i, x) : i \in \Field_p^{d+3}, x \in X'_i}$} and similarly $U'$ and $V'$, where $X'_i, U'_i, V'_i \subseteq \Field_p^d$ are the sets defined as follows:
\begin{alignat*}{3}
    &X'_{a + b} &&= Y, \\
    &X'_i &&= \smash{\Field_p^d} &&\qquad(\text{for $i \in W \setminus \set{a + b}$}), \\
    &X'_{i + w} &&= \smash{\Field_p^d} &&\qquad(\text{for $i \in W$}), \\
    &X'_{i - w} &&= \smash{\Field_p^d} &&\qquad(\text{for $i \in Z \cup \set{a}$, only relevant if $p > 2$})
\intertext{and}
    &U'_a &&= U, \\
    &U'_z &&= \set{z} &&\qquad(\text{for $z \in Z$}), \\
    &U'_{i + w} &&= \smash{\Field_p^d} &&\qquad(\text{for $i \in W \setminus \set{a + b}$}),
\intertext{and}
    &V'_b &&= V, \\
    &V'_{a + b - z} &&= \set{0} &&\qquad(\text{for $z \in Z$}), \\
    &V'_{-w} &&= \smash{\Field_p^d}.
\end{alignat*}
Moreover, all sets $X'_i, U'_i, V'_i$, which we have not specified here, are set to be empty.

\paragraph{Soundness}
For the soundness, assume that there is a pair $A \subseteq U, B \subseteq V$ with $X \subseteq A + B \subseteq Y$. We construct sets $A' \subseteq U'$ and $B' \subseteq V'$ with $X' = A' + B'$. Write $A' = \set{(i, x) : i \in I, x \in A'_i}$ and similarly for $B'$ where the sets $A'_i, B'_i \subseteq \Field_p^d$ are defined as follows:
\begin{alignat*}{3}
    &A'_a &&= A, \\
    &A'_z &&= \set{z} &&\qquad(\text{for $z \in Z$}), \\
    &A'_{i + w} &&= \Field_p^d &&\qquad(\text{for $i \in W \setminus \set{a + b}$}),
\intertext{and}
    &B'_b &&= B, \\
    &B'_{a + b - z} &&= \set{0} &&\qquad(\text{for $z \in Z$}), \\
    &A'_{-w} &&= \Field_p^d.
\end{alignat*}
Again, all sets $A'_i, B'_i$ that we have not specified here are assumed to be empty. From the construction, it is immediate that $A' \subseteq U'$ and that $B' \subseteq V'$.

Next, we show that $A' + B' \subseteq X'$. To this end, we verify that $A'_i + B'_j \subseteq X'_{i + j}$ for all $i, j \in I$. This is trivial whenever~\smash{$X'_{i + j} = \Field_p^d$} which covers almost all cases. The only other relevant cases are as follows:
\begin{alignat*}{3}
    &A'_a + B'_b &&= A + B \subseteq Y = X'_{a + b}, \\
    &A'_z + B'_{a + b - z} &&= \set{z} + \set{0} \subseteq Z \subseteq Y = X'_{a + b}.
\end{alignat*}

Let us finally check that $X' \subseteq A' + B'$. That is, we check that for all $k$, $X'_k \subseteq \bigcup_{i + j = k} (A'_i + B'_j)$. The most interesting case is
\begin{equation*}
    X'_{a + b} = Y = X \cup Z \subseteq (A + B) \cup \bigcup_{z \in Z} \set{z} = (A'_a + B'_b) \cup \bigcup_{z \in Z} (A'_z + B'_{a + b - z}).
\end{equation*}
The less interesting cases are
\begin{alignat*}{5}
    &X'_i &&= \smash{\Field_p^d} &&= \smash{\Field_p^d} + \smash{\Field_p^d} &&= A'_{i + w} + B'_{-w} &&\qquad(\text{for $i \in W \setminus \set{a + b}$}), \\
    &X'_{i + w} &&= \smash{\Field_p^d} &&= \smash{\Field_p^d} + B &&= A'_{i + w - b} + B'_b &&\qquad(\text{for $i \in W$}), \\
    &X'_{z - w} &&= \smash{\Field_p^d} &&= \set{z} + \smash{\Field_p^d} &&= A'_z + B'_{-w} &&\qquad(\text{for $z \in Z$}), \\
    &X'_{a - w} &&= \smash{\Field_p^d} &&= A + \smash{\Field_p^d} &&= A'_a + B'_{-w}.
\end{alignat*}
(Strictly speaking, here we assumed that $A, B \neq \emptyset$. This can easily be enforced by a trivial corner case.)

\paragraph{Completeness}
Suppose that there are sets~\makebox{$A' \subseteq U', B' \subseteq V'$} with $X' = A' + B'$; we show that there are sets $A \subseteq U, B \subseteq V$ with $X \subseteq A + B \subseteq Y$. We again write~\smash{$A'_i = \set{x \in \Field_p^d : (i, x) \in A'}$} for~\smash{$i \in \Field_p^{d+3}$}, and similarly for $B'$. From our construction, it is clear that
\begin{equation*}
    X'_k = \bigcup_{\substack{i, j \in \Field_p^{d+3}\\i + j = k}} (A'_i + B'_j) \subseteq \bigcup_{\substack{i, j \in \Field_p^{d+3}\\i + j = k}} (U'_i + V'_j).
\end{equation*}
In particular, consider the case $k = a + b$. By carefully checking all combinations of $i$ and $j$ above, we in fact have that
\begin{equation*}
    X'_{a + b} = (A'_a + B'_b) \cup \bigcup_{z \in Z} ((A'_z + B'_{a + b - z}) \cup (A'_{a + b - z} + B'_{z})).
\end{equation*}
By the construction of $U'$ and $V'$, the big union on the right becomes exactly a subset $Z' \subseteq Z$. Therefore, and using that $X'_{a + b} = Y$, we have that $Z' \cup (A'_a + B'_b) = Y$. In particular, $X \subseteq A'_a + B'_b$, and we have thus witnessed sets $A = A'_a$ and $B = B'_b$ as desired.
\end{proof}

\section{Reduction from 3-SAT} \label{sec:sat} 
In this section, we employ the technical lemmas from Sections \ref{sec:pos} and \ref{sec:masking} to show a polynomial-time reduction from 3-SAT to Sumset Recognition. We start with the integer case in \cref{sec:sat:sec:int} and then deal with finite fields in \cref{sec:sat:sec:Fp}.

\subsection{Reduction for \texorpdfstring{$\Int$}{Z}} \label{sec:sat:sec:int}
Given a 3-SAT instance, a formula $\phi$ on $n$ variables and $m$ clauses, the reduction outputs a Sumset Recognition instance; a set $X_\phi \subseteq [0,2^{57}(n+m)^4]$, 
such that $\phi$ is satisfiable if an only if there exists $A \subseteq \mathbb{Z}$ such that $X_{\phi}=A+A$. 
The main technical part in this section is a reduction for a related problem where given $\phi$, we output three sets $X,Y,U$ such that $\phi$ is satisfiable if and only if 
exists $A \subseteq U$ such that $X \subseteq A+A \subseteq Y$. In this reduction, we will employ the following Sidon set construction. Recall that we call a set $S$ \emph{Sidon} if it does not contain nontrivial solutions to the equation $i + j = k + \ell$, or equivalently, if~\smash{$|S + S| = \binom{|S|+1}{2}$}.

\begin{lemma}[Sidon Sets over $\Int$] \label{lem:sidon-int}
    In time $\poly(n)$ we can compute a Sidon set $S \subseteq [0, 4n^2]$ of~size~$n$.
\end{lemma}
\begin{proof}
    This result is classic, but we include a quick proof based on Erdős and Tur{\'a}n's construction~\cite{ErdosT41} for the sake of completeness. By Bertrand's postulate, there is a prime number $n \leq p < 2n$, and we can find $p$ by, say, Eratosthenes' sieve in deterministic time $\Order(n \log \log n)$. As Erdős and Tur{\'a}n proved~\cite{ErdosT41}, the set
    \begin{equation*}
        S' = \set{(x, x^2) : x \in \Field_p}
    \end{equation*}
    is a Sidon set over $\Field_p^2$. Indeed, suppose that $(i, i^2) + (j, j^2) = (k, k^2) + (\ell, \ell^2)$. Equivalently, we can write $i - k = \ell - j$ and $i^2 - k^2 = \ell^2 - j^2$. Thus, either $i - k = \ell - j = 0$ or $i + k = \ell + j$. In the latter case, in combination with the initial assumption, we find that $i = \ell$ and $j = k$. This set~$S'$ can be lifted to the integers by embedding $\Field_p^2$ into the integers via $(x, y) \mapsto x \cdot 2p + y$ (viewing $0 \leq x, y < p$ as integers). The resulting set $S$ satisfies that $S \subseteq [0, 2p(p-1) + 1] \subseteq [0,4n^2]$ and can clearly be constructed in time $\poly(n)$.
\end{proof}

\begin{corollary} \label{cor:sidon-modified}
    In time $\poly(n, m)$ we can compute integers~\makebox{$s_1, \dots, s_n, t_1, \dots, t_m$} such that:
    \begin{enumerate}[label=(\roman*)]
        \item $0 \leq s_1, \dots, s_n \leq N$ and $4N \leq t_1, \dots, t_m \leq  5N$, where $N = 16(n + m)^2$.
        \item Whenever $|2 s_i - s_j - s_k| < 4$, then $i = j = k$.
        \item Whenever $|s_i - s_j + t_k - t_\ell| < 4$, then $i = j$ and $k = \ell$.
    \end{enumerate}
\end{corollary}
\begin{proof}
    The previous lemma yields a Sidon set $\set{s_1', \dots, s_n', t_1', \dots, t_m'} \subseteq [0,4(n + m)^2]$; we then choose $s_i := 4 s_i'$ and $t_i := 4t_i' + 64(n + m)^2$. Property~(i) is fulfilled by construction. For property~(ii), note that $|2 s_i - s_j - s_k| < 4$ if and only if $2 s_i' - s_j' - s_k' = 0$. The Sidon assumption guarantees that this happens if and only if $i = j = k$. Property~(iii) follows by a similar argument.
\end{proof}

\begin{lemma}[Reduction from 3-SAT for $\Int$]\label{lem:3sat-int}
    There is an algorithm that, given a 3-SAT formula $\phi$ on $n$ variables and $m$ clauses, outputs sets $X, Y, U \subseteq [0,160(n+m)^2]$, 
    such that
    \begin{equation*}
        \text{$\phi$ is satisfiable} \qquad\text{if and only if}\qquad \exists A \subseteq U : X \subseteq A + A \subseteq Y.
    \end{equation*}
    The algorithm runs in time $O(n^2 + m^2)$. 
\end{lemma}

\begin{proof}
    We denote the variables of $\phi$ by $x_1,\ldots,x_n$. 
    Let $s_1, \dots, s_n, t_1, \dots, t_m$ denote the integers as constructed by \cref{cor:sidon-modified}.
    We then define the three sets as follows. Let $X := \{t_1,\ldots,t_m\}$, and let
    \begin{alignat*}{2}
        C_k &:= \bigcup_{\substack{x_i \text{ appears} \\ \text{positively in} \\ \text{the $k$-th clause.}}} \set{t_k - s_i - 1} \cup \bigcup_{\substack{x_i \text{ appears} \\ \text{negatively in} \\ \text{the $k$-th clause.}}} \set{t_k - s_i} &&\qquad(\text{for $k \in [m]$}), \\
        U &:= \bigcup_{i \in [n]} (s_i + \set{0, 1}) \cup \bigcup_{\substack{k \in [m]}} C_k, \\
        Y &:= \bigcup_{i \in [n]} (2s_i + \set{0, 2}) \cup \bigcup_{\substack{i, j \in [n]\\i \neq j}} (s_i + s_j + \set{0, 1, 2}) \\
        &\qquad\qquad\cup \bigcup_{\substack{i \in [n]\\k \in [m]}} (s_i + \set{0, 1} + C_k) \cup \bigcup_{\substack{k, \ell \in [m]}} (C_k + C_{\ell}).
    \end{alignat*}
    \cref{cor:sidon-modified}~(i) guarantees that $0 \leq s_1, \dots, s_n \leq N$ and $4N\leq t_1, \dots, t_m \leq 5N$ where $N = 16(n + m)^2$. Hence, all sets contain only nonnegative integers of size at most $10N \leq 160(n + m)^2$.

    \paragraph{Soundness.}
        Suppose that $\phi$ is satisfiable and let $\alpha \in \set{0, 1}^n$ denote a satisfying assignment of~$\phi$. We show that there is a set $A \subseteq U$ satisfying that $X \subseteq A + A \subseteq Y$. We choose
        \[
            A = \bigcup_{i \in [n]} \set{s_i + \alpha_i} \cup \bigcup_{k \in [m]} C_k. 
        \]
        Observe that $A \subseteq U$. In addition, it is not too difficult to see that $A + A \subseteq Y$, since: 
        \begin{align*}
            A + A &= \bigcup_{i, j \in [n]} \set{s_i + s_j + \alpha_i + \alpha_j} \cup \bigcup_{\substack{i \in [n]\\k \in [m]}} (s_i + \alpha_i + C_k) \cup \bigcup_{k, \ell \in [m]} (C_k + C_\ell).
        \end{align*}
        The second and third terms are clearly included in the third and fourth terms of $Y$. For the first term we distinguish two cases: If $i = j$, then $\set{s_i + s_j + \alpha_i + \alpha_j} \subseteq 2s_i + \set{0, 2}$ which is included in the first term of $Y$. Otherwise, if $i \neq j$, then $\set{s_i + s_j + \alpha_i + \alpha_j} \subseteq s_i + s_j + \set{0, 1, 2}$ is clearly included in the second term of $Y$.

        Let us finally prove that $X \subseteq A + A$. As $X = \set{t_1, \dots, t_m}$ we need to show, for each $k \in [m]$, that $t_k \in A + A$.
        Since $\alpha$ satisfies $\phi$, it satisfies some literal of some variable $x_i$ in the $k$-th clause.  
        Consider the term $s_i + \alpha_i + C_k \subseteq A+A$. Assume, on the one hand, that $x_i$ appears positively in this clause, i.e., $\alpha_i = 1$. 
        Then by the definition of $C_k$ we have $t_k = s_i + 1 + (t_k - s_i - 1) \in A+A$. On the other hand, if $x_i$ appears negatively in $C_k$ (and thus $\alpha_0 = 0$), then $t_k = s_i + 0 + (t_k - s_i) \in A + A$. 

    \paragraph{Completeness}
        Before we start with the actual completeness proof, let us first assert that our construction admits the following properties:
        \begin{itemize}
            \item $2s_i + 1 \not\in Y$. Of course $2s_i + 1$ does not appear in $2s_i + \set{0, 2}$, but it requires an argument that $2s_1 + 1$ does not spuriously appear in the other terms in $Y$: Since $Y \subseteq U + U$, any other expression would lead to a nontrivial solution of $2s_i = s_j + s_k \pm 3$, or $2s_i = s_j + t_k - s_\ell \pm 3$ or $2s_i = t_k - s_j + t_\ell - s_h \pm 3$. The first equation is ruled out by~\cref{cor:sidon-modified}~(ii), and the latter two are ruled out by~\cref{cor:sidon-modified}~(i).
            \item The only possible representations of $t_k$ in $U + U$ are $t_k = (s_i + 0) + (t_k - s_i)$ and $t_k = (s_i + 1) + (t_k - s_i - 1)$ (whenever the variable~$x_i$ appears in the $k$-th clause). Indeed, any other representation would take one of the forms $t_k = s_i + s_j \pm 3$, $t_k = s_i + (t_\ell - s_j) \pm 3$ (for~\makebox{$i \neq j$} or $k \neq \ell$), or~\smash{$t_k = (t_\ell - s_i) + (t_h - s_j) \pm 3$}. The first and last equations are ruled out by~\cref{cor:sidon-modified}~(i), and the second one is ruled out by~\cref{cor:sidon-modified}~(iii).
        \end{itemize}

        Equipped with these insights, we continue with the completeness proof. So, assume that there is a set $A \subseteq U$ satisfying that $X \subseteq A + A \subseteq Y$; we construct a satisfying assignment of $\phi$. The first claim implies that it cannot simultaneously happen that $s_i + 0 \in A$ and that $s_i + 1 \in A$ (as otherwise~\makebox{$A + A \not\subseteq Y$}). In light of this, we let $\alpha_i = 1$ if and only if $s_i + 1 \in A$. To prove that $\alpha$ is satisfying, consider any clause $k$. Since $X \subseteq A + A$, we in particular have that~\makebox{$t_k \in A + A \subseteq U + U$}. By the second claim, $t_k$ must be represented either in the form of $t_k = (s_i + 0) + (t_k - s_i)$ or~\makebox{$t_k = (s_i + 1) + (t_j - s_i - 1)$}. In the former case we have that $s_i + 0 \in A$ (and thus $\alpha_i = 0$) and $t_k - s_i \in U$ (and thus $x_i$ appears negatively in the $k$-th clause). In the latter case we have that $s_i + 1 \in A$ (and thus $\alpha_i = 1$) and~\makebox{$t_k - s_i - 1 \in U$} (and thus $x_i$ appears positively in the $k$-th clause). In both cases, the clause is satisfied.
    \end{proof}

    \begin{proof}[Proof of Theorem \ref{thm:np-hard-integer}]
        Let us now show a polynomial-time reduction from 3-SAT to the Sumset Recognition problem. Since the problem is clearly in NP, this will prove that the problem is NP-complete. 
        Given a 3-SAT instance; a formula $\phi$ on $n$ variables and $m$ clauses, apply the algorithm from \cref{lem:3sat-int} to obtain three sets $X,Y,U \subseteq [0,160(n+m)^2] \subseteq [0,2^8(n+m)^2]$, such that: 
        \begin{equation*}
            \text{$\phi$ is satisfiable} \qquad\text{if and only if}\qquad \exists A \subseteq U : X \subseteq A + A \subseteq Y.
        \end{equation*}
        Then, apply the algorithm from \cref{lem:masking-int} on $X,Y,U$ to obtain two sets $U',X' \subseteq [0,2^{17}\cdot (2^8(n+m)^2)^2] = [0,2^{33}(n+m)^4]$, such that: 
        \begin{equation*}
            \exists A \subseteq U : X \subseteq A + A \subseteq Y \qquad\text{if and only if}\qquad \exists A' \subseteq U' : X' = A' + A'.
        \end{equation*}
        As a corollary, $\phi$ is satisfiable if and only if $\exists A' \subseteq U'$ such that $X' = A'+A'$.
        Finally, apply the algorithm from \cref{lem:pos-int} on $X',U'$ to obtain a Sumset Recognition instance $X_{\phi} \subseteq [0,2^{24} \cdot 2^{33} (n+m)^4] = [0,2^{57}(n+m)^4]$, such that: 
        \begin{equation*}
            \exists A' \subseteq U' : X' = A' + A' \qquad\text{if and only if}\qquad \exists A'' \subseteq \Int : X_{\phi} = A'' + A''.
        \end{equation*}
        Thus, $\phi$ is satisfiable if and only if $\exists A'' \subseteq \mathbb{Z}$ such that $A''+A'' = X_{\phi}$.

        The running times of the above algorithms are all $\poly(n,m)$. Hence, we showed a polynomial-time reduction from 3-SAT to Sumset Recognition. 
    \end{proof}

\subsection{Reduction for \texorpdfstring{$\Field_p^d$}{Fpd}} \label{sec:sat:sec:Fp}
In this section, we rework the reduction for finite fields. We start with a similar construction of Sidon sets over $\Field_p^d$.

\begin{lemma}[Sidon Sets over~\smash{$\Field_p^d$}] \label{lem:sidon-Fp}
Let $p$ be a prime, and let $d = 2\ceil{\log_p(n)}$. In time $\poly(n)$ we can construct a Sidon set~\smash{$S \subseteq \Field_p^d$} of size $n$.
\end{lemma}
\begin{proof}
We follow a classic construction dating back to Bose and Ray-Chaudhuri~\cite{BoseR60}; see also~\cite{Cilleruelo12}. Let $q = p^d$. By representing the finite field $\Field_q$ as $\Field_p^d$, it suffices to construct a Sidon set in~$\Field_q^2$. For odd $p$, such a set is for instance $S = \set{(x, x^2) : x \in \Field_q}$; the proof is identical to the argument in \cref{lem:sidon-int}. We take $S = \set{(x, x^3) : x \in \Field_q}$ if instead $p = 2$~\cite{BoseR60}. Polynomial arithmetic can be implemented in time $\poly(q)$, and so the construction time of $S$ is also bounded by $\poly(n)$.
\end{proof}

\begin{lemma}[Reduction from 3-SAT for~\smash{$\Field_p^d$}] \label{lem:sat-Fp}
    There is an algorithm that, given a 3-SAT formula~$\phi$ on $n$ variables and $m$ clauses, outputs sets~\smash{$X, Y, U, V \subseteq \Field_p^d$}, where~\smash{$d = 2 \ceil{\log_p(n + m)} + 5$}, such that
    \begin{equation*}
        \text{$\phi$ is satisfiable} \qquad\text{if and only if}\qquad \exists A \subseteq U, B \subseteq V : X \subseteq A + B \subseteq Y.
    \end{equation*}
    The algorithm runs in time $\poly(n, m)$. 
\end{lemma}
\begin{proof}
Let $x_1, \dots, x_n$ denote the variables in $\phi$. Further, let $a_1, \dots, a_n, b_1, \dots, b_n, t_1, \dots, t_m \in \Field_p^d$, where $d = 2\ceil{\log_p(2n + m)} + 3 \leq 2 \ceil{\log_p(n + m)} + 5$, denote the elements of a Sidon set of size~\makebox{$2n + m$} as constructed by \cref{lem:sidon-Fp}. We use the three remaining dimensions to obtain pairwise orthogonal nonzero vectors $v_0, v_1, w$ that are also orthogonal to $a_1, \dots, a_n, b_1, \dots, b_n, t_1, \dots, t_m$. We then define:
\begin{alignat*}{2}
    C_k &:= \bigcup_{\substack{x_i \text{ appears} \\ \text{positively in} \\ \text{the $k$-th clause.}}} \set{t_k - a_i - v_1} \cup \bigcup_{\substack{x_i \text{ appears} \\ \text{negatively in} \\ \text{the $k$-th clause.}}} \set{t_k - a_i - v_0}, \\
    U &:= \bigcup_{i \in [n]} (a_i + \set{v_0, v_1}), \\
    V &:= \bigcup_{i \in [n]} (b_i - \set{v_0, v_1}) \cup \bigcup_{\substack{k \in [m]}} (w + C_k), \\
    X &:= \bigcup_{i \in [n]} \set{a_i + b_i} \cup \bigcup_{k \in [m]} \set{w + t_k}, \\
    Y &:= \bigcup_{i \in [n]} \set{a_i + b_i} \cup \bigcup_{\substack{i, j \in [n]\\i \neq j}} (a_i + b_j + \set{v_0, v_1} - \set{v_0, v_1}) \cup \bigcup_{\substack{i \in [n]\\k \in [m]}} (a_i + \set{v_0, v_1} + w + C_k).
\end{alignat*}

\paragraph{Soundness.}
    Suppose that $\phi$ is satisfiable and let $\alpha \in \set{0, 1}^n$ denote a satisfying assignment of~$\phi$. We show that there are sets $A \subseteq U$ and $B \subseteq V$ satisfying that $X \subseteq A + B \subseteq Y$. To this end, we pick
    \begin{align*}
        A &= \bigcup_{i \in [n]} \set{a_i + v_{\alpha_i}}, \\
        B &= \bigcup_{i \in [n]} \set{b_i - v_{\alpha_i}} \cup \bigcup_{k \in [m]} (w + C_k).
    \end{align*}
    Observe that $A \subseteq U$ and $B \subseteq V$, and that
    \begin{equation*}
        A + B = \bigcup_{i, j \in [n]} \set{a_i + b_j + v_{\alpha_i} - v_{\alpha_j}} \cup \bigcup_{\substack{i \in [n]\\k \in [m]}} (a_i + v_{\alpha_i} + w + C_k) \subseteq Y.
    \end{equation*}
    We finally show that $X \subseteq A + A$. On the one hand, it is clear that $a_i + b_i \in A + B$ for all $i \in [n]$. On the other hand, focus on an arbitrary clause $k \in [m]$. Since $\alpha$ is a satisfying assignment, there must be a variable $x_i$ which satisfies the $k$-th clause. If the variable appears positively (and thus~\makebox{$\alpha_i = 1$}), then $w + t_k = (a_i + v_1) + (w + t_k - a_i - v_1) \in A + B$. Otherwise, if the variable appears negatively (and thus $\alpha_i = 0$), we find $w + t_k = (a_i + v_0) + (w + t_k - a_i - v_0) \in A + B$.

\paragraph{Completeness}
    For the completeness proof, we again first gather some structural observations. Specifically:
    \begin{itemize}
        \item $a_i + b_i + v_0 - v_1 \not\in Y$ and $a_i + b_i + v_1 - v_0 \not\in Y$. We check that neither of these elements appears spuriously in the second and third terms in $Y$. For the third term, this is clear as it contains the shift $+ w$, which is orthogonal to $a_i + b_i \pm v_0 \pm v_1$. And if the second term would contain, say, $a_i + b_i + v_0 - v_1$, then we would also obtain a solution to the equation~\makebox{$a_i + b_i = a_j + b_h$} (where $j \neq h$). This is ruled out since $a_1, \dots, a_n, b_1, \dots, b_n$ stem from a common Sidon set. The same argument shows that $a_i + b_i$ can be represented in exactly two ways in~\makebox{$U + V$}, namely as $a_i + b_i = (a_i + v_0) + (b_i - v_0)$ and $a_i + b_i = (a_i + v_1) + (b_1 - v_1)$.
        \item The only possible representations of $w + t_k$ in $U + V$ are $t_k = (a_i + v_0) + (w + t_k - a_i - v_0)$ and $t_k = (a_i + v_1) + (w + t_k - a_i - v_1)$ (whenever the variable~$x_i$ appears in the $k$-th clause). Indeed, any other representation would necessarily take the form $w + t_k = a_i + b_j$ or $w + t_k = a_i + (w + t_\ell - a_j)$ (where $i \neq j$ or $k \neq \ell$). The first equation is ruled out since $w$ is orthogonal to all other terms in the equation, and the second out is ruled out since $a_1, \dots, a_n, t_1, \dots, t_m$ stem from a common Sidon set.
    \end{itemize}

    We are ready to complete the proof. So assume that there are sets $A \subseteq U$ and $B \subseteq V$ satisfying that $X \subseteq A + B \subseteq Y$; we construct a satisfying assignment of $\phi$. The first claim implies that either~\makebox{$a_i + v_0 \in A$} or that $a_i + v_1 \in A$. Indeed, if both elements were missing we would violate the condition that $a_i + b_i \in X \subseteq A + B$ (which requires at least one of these two elements to be present). Suppose instead that both elements are present in $A$. Then the condition $a_i + b_i \in X \subseteq A + B$ entails that $b_i - v_0 \in B$ or $b_i - v_1 \in B$. Either case leads to a contradiction as then also $a_i + b_i + v_1 - v_0 \in A + B$ or $a_i + b_i + v_0 - v_1 \in A + B$, respectively, both of which contradict $A + B \subseteq Y$ by the first claim.
    
    Therefore, let $\alpha_i = 1$ if and only if $a_i + v_1 \in A$. To prove that $\alpha$ is satisfying, consider any clause $k$. As $X \subseteq A + B$, in particular we have that~\makebox{$w + t_k \in A + B \subseteq U + V$}. By the second claim, $t_k$ must be represented either in the form of $t_k = (a_i + v_0) + (w + t_k - a_i - v_0)$ or~\makebox{$t_k = (a_i + v_1) + (w + t_k - a_i - v_1)$}. In the former case we have that $a_i + v_0 \in A$ (and thus~\makebox{$\alpha_i = 0$}) and $w + t_k - a_i - v_0 \in V$ (and thus $x_i$ appears negatively in the $k$-th clause). In the latter case we have that $a_i + v_1 \in A$ (and thus $\alpha_i = 1$) and~\makebox{$w + t_k - a_i - v_1 \in V$} (and thus $x_i$ appears positively in the $k$-th clause). In both cases, the clause is satisfied.
\end{proof}

\begin{proof}[Proof of \cref{thm:np-hard-finite-field}]
This proof is analogous to the proof of \cref{thm:np-hard-integer}, and for this reason, we will be brief. We concatenate the three reductions in \cref{lem:sat-Fp,lem:masking-Fp,lem:pos-Fp}. Specifically, by \cref{lem:sat-Fp} we transform a given 3-SAT instance $\phi$ on $n$ variables and $\poly(n)$ clauses into sets~\smash{$X, Y, U, V \subseteq \Field_p^d$}, where~\smash{$d = \Order(\log_p n)$}, which is then transformed by \cref{lem:masking-Fp} into sets~\smash{$X', U', V' \subseteq \Field_p^{d'}$}, where~\smash{$d' = \Order(\log_p n)$}, which are finally transformed by \cref{lem:pos-Fp} into a set~\smash{$X'' \subseteq \Field_p^{d''}$}, where~\smash{$d'' = \Order(\log_p n)$}. \cref{lem:sat-Fp,lem:masking-Fp,lem:pos-Fp} guarantee that $\phi$ is satisfiable if and only if $\exists A \subseteq U, B \subseteq V$ with $X \subseteq A + B \subseteq Y$ if and only if $\exists A' \subseteq U', B' \subseteq V'$ with~\smash{$X' = A' + B'$} if and only if~\smash{$\exists A'' \subseteq \Field_p^{d''}$} with $X'' = A'' + A''$. We have thereby transformed the initial 3-SUM instance into an equivalent instance of the Sumset Recognition problem over~\smash{$\Field_p^{d''}$}, where~\smash{$|\Field_p^{d''}| = p^{d''} = \poly(n)$}. The total running time is $\poly(n)$.
\end{proof}

\subsection{ETH-Hardness}
Finally, we analyze our hardness reduction through a more fine-grained lens and derive exact conditional lower bounds for the Sumset Recognition problem based on the \emph{Exponential Time Hypothesis~(ETH)}. This hypothesis, as introduced by Impagliazzo and Paturi~\cite{ImpagliazzoP01}, postulates that the 3-SAT problem cannot be solved in time $2^{\epsilon n} \poly(n)$ for all $\epsilon > 0$. We rely on the well-established Sparsification Lemma to better control the number of clauses in a 3-SAT instance:

\begin{lemma}[Sparsification Lemma \cite{ImpagliazzoPZ01}]
    \label{lem:sparsification}
    For any $\eps > 0$, there exists $c_\eps > 0$ and an algorithm that, given a 3-SAT formula $\phi$ on $n$ variables, computes in $2^{\eps n} \poly(n)$ time 3-SAT instances $\phi_1, \ldots , \phi_k$ with $k \leq 2^{\eps n}$ such that $\phi$ is satisfiable if and only if $\bigvee_{i=1}^k \phi_i$ is satisfiable. Moreover, the number of clauses in each formula is at most $c_{\eps} n$. 
\end{lemma}

\begin{proof}[Proof of \cref{thm:eth-hard-integer}]
We show that if there is an algorithm for Sumset Recognition running in time $2^{o(n^{1/4})}$, then there is also an algorithm for 3-SAT running in time $2^{o(n)}$, thus refuting the Exponential Time Hypothesis. 
More formally, assume, for the sake of contradiction, that for any~$\delta > 0 $, there exists an algorithm for Sumset Recognition running in time \smash{$2^{\delta N^{1/4}} \poly(N)$} 
on instances taken from $[0,N]$. Let us show that for any $\eps > 0$, there is an algorithm for 3-SAT whose running time is $2^{\eps n} \poly(n)$. 

Let $\phi$ be a 3-SAT formula on $n$ variables and let $\eps > 0$. 
Apply the algorithm from \cref{lem:sparsification} with $\eps' := \eps/2$ to obtain $k \leq 2^{\eps' n}$ formulas $\phi_{1},\ldots,\phi_k$ such that each formula has 
at most $c_{\eps'} n$ clauses, and $\phi$ is satisfiable if and only if $\bigvee_{i=1}^k \phi_i$ is satisfiable. Its running time is $2^{\eps' n} \poly(n)$. 
Now, apply the reduction from 3-SAT to Sumset Recognition, given in the proof of \cref{thm:np-hard-integer}, to each of the formulas. 
The result of this step is $k$ instances of Sumset Recognition~\makebox{$X_1, \ldots , X_k \subseteq [0,N]$}, where~\makebox{$N \leq 2^{60} (n+c_{\eps'} n)^4 = 2^{60}(1+c_{\eps'})^4 n^4$}. 
Moreover, $\phi$ is satisfiable if and only if one of the Sumset Recognition instances has a sumset root. 
Hence, to decide if $\phi$ is satisfiable, apply the algorithm for Sumset Recognition on each of those instances. The total running time is then 
$2^{\eps' n} \cdot \poly(n) +  k \cdot 2^{\delta (2^{60}(1+c_{\eps'})^4 n^4)^{1/4} } \cdot \poly(N) = 2^{\eps' n} \cdot 2^{\delta 2^{15}(1+c_{\eps'}) n} \cdot \poly(n)$. 
By setting $\delta \leq \eps'/(2^{15} (1+c_{\eps'}))$, the running time becomes $2^{2 \eps' n} \cdot \poly(n) = 2^{\eps n} \cdot \poly(n)$. 
\end{proof}

By exactly the same proof, we can also establish that it is ETH-hard to solve the Sumset Recognition problem over a finite field $\Field$ in time~\smash{$2^{\order(|\Field|^{1/4})}$}.

\bibliographystyle{plainurl}
\bibliography{paper}

\end{document}